\documentclass[11pt]{article}

\usepackage{graphicx}
\usepackage{amssymb, amsmath, amsthm}
\usepackage{color}
\usepackage[margin=1in]{geometry}
\usepackage{pdfpages}
\usepackage{enumerate}

\usepackage {tikz}
\definecolor {processblue}{cmyk}{0.96,0,0,0}

\usepackage{algorithm}
\usepackage{algorithmic}

\newtheorem{theorem}{Theorem}

\newtheorem{lemma}[theorem]{Lemma}
\newtheorem{claim}[theorem]{Claim}
\newtheorem{corollary}[theorem]{Corollary}

\newcommand{\cA}{{\cal A}}

\newcommand{\cI}{{\cal I}}

\newcommand{\cL}{{\cal L}}

\newcommand{\cP}{{\cal P}}

\newcommand{\cT}{{\cal T}}

\begin{document}

\title{
Efficient Approximations for the Online Dispersion Problem}

\author{Jing Chen \hspace{30pt} Bo Li \hspace{30pt} Yingkai Li \\
Department of Computer Science, Stony Brook University\\
Stony Brook, NY 11794\\
\texttt{\{jingchen, boli2, yingkli\}@cs.stonybrook.edu}}

\maketitle

\begin{abstract}
The {\em dispersion} problem has been widely studied in computational geometry and facility location, and is closely related to the packing problem.
The goal is to locate $n$ points (e.g., facilities or persons) in a $k$-dimensional polytope, so that {\em they are far away from each other
and from the boundary of the polytope.}
In many real-world scenarios however,
the points arrive and depart at different times,
 and decisions must be made without knowing future events.
Therefore we study,
for the first time in the literature,
the {\em online dispersion} problem in Euclidean space.

There are two natural objectives when time is involved:
the {\em all-time worst-case} (ATWC) problem tries to  maximize the minimum distance that ever appears
at any time;
and the {\em cumulative distance} (CD) problem tries to maximize
 the integral of the minimum distance throughout the whole time interval.
Interestingly,
the online  problems are highly non-trivial even on a segment.
For cumulative distance, this remains the case even when the problem is time-dependent but offline,
with all the arriving and departure times given in advance.

For the online ATWC problem on a segment, we
construct a deterministic polynomial-time algorithm
which is
$(2\ln 2+\epsilon)$-competitive, where
$\epsilon>0$ can be arbitrarily small and the algorithm's running time is polynomial in $\frac{1}{\epsilon}$.
We show this algorithm is actually
{\em optimal}.
For the same problem in a square, we provide a $1.591$-competitive algorithm and a $1.183$ lower-bound.
Furthermore, for arbitrary $k$-dimensional polytopes with $k\geq 2$,
we provide a $\frac{2}{1-\epsilon}$-competitive algorithm and
a~$\frac{7}{6}$ lower-bound.
All our lower-bounds come from the structure of the online problems and
hold even when computational complexity is not a concern.
Interestingly,
for the offline CD problem in arbitrary $k$-dimensional polytopes,
we provide a polynomial-time black-box reduction to the online ATWC problem,
and the resulting competitive ratio increases by a factor of at most~2.
Our techniques also apply to online dispersion problems with different boundary conditions.

\bigskip


\end{abstract}

\section{Introduction}

The problem of assigning elements to locations in a given area comes up only too often
in real life: where to seat the customers in a restaurant,
where to put certain facilities in a city,
where to build nuclear power stations
in a country, etc.
Different problems have different features and constraints,
but one common feature that appears in many of them is {\em
not to locate the elements too close to each other}:
for people's privacy, for  environmental safety, and/or for serving more users.
Another feature is also common in many applications: that is,
{\em not
to locate the elements too close to the boundary of the area.} Indeed,
for security reasons,
important national industrial facilities in many countries are built
at a safe distance away from the border.
Such problems have been widely studied in computational geometry and facility location; see, e.g.,
\cite{ravi1994heuristic, ben2000obnoxious, baur2001approximation, benkert2006polynomial}.
In particular, in the {\em dispersion} problem defined by \cite{baur2001approximation},
there is
a $k$-dimensional polytope $P$ and an integer $n$,
and the goal is to locate $n$ points in $P$ so as to maximize
{\em the minimum distance among them and from them to the boundary of $P$.}

%

However, there is another important feature in all the scenarios mentioned above
and many other real-world scenarios:
the presence of elements is {\em time-dependent}
and decisions need to be made along time, without knowing
when the elements will come and go in the future.
Indeed, it may be hard to move an element once it is located,
making it infeasible for the decision maker to
relocate all the present elements according to
the optimal static solution
when an arrival/departure event occurs.
Online dispersion and facility location problems have been studied when the underlying locations
are vertices of a graph \cite{meyerson2001online, fotakis2008competitive, mettu2003online}.
In this paper we consider, for the first time in the literature, the
online dispersion problem in Euclidean space.
The arriving and departure times of points are chosen by an adversary who
knows everything and works adaptively. An online dispersion algorithm decides where to locate a point upon its arrival, without any knowledge about future events.

\subsection{Main Results}

We focus on two natural objectives for the online problem:
the {\em all-time worst-case} (ATWC) problem, which aims at maximizing the minimum distance that ever appears
at any time; and the {\em cumulative distance} (CD) problem, which aims at maximizing the integral of the minimum distance throughout the whole time interval.
Although polynomial-time constant approximations
have been given when time is not involved \cite{baur2001approximation, benkert2006polynomial},
nothing was known about the online problem.
As we will show, solutions for the online problem
are already
complex even on a segment.
For cumulative distance,
even when the problem is time-dependent but offline,
with all the arriving and departure times given in advance, it remains unclear how
 to efficiently compute the optimal solution.
We formally define the problem in Section~\ref{sec:preliminary}
and summarize our results in Table~\ref{table:summary} below.

The most technical parts are the online ATWC problem and the offline time-dependent
CD problem. Interestingly, we provide an efficient reduction from
the offline CD problem to the online ATWC problem, and
 show that in order to solve the former, one
  can use an algorithm for the latter as a black-box.
For the online ATWC problem, it is not hard to see that a natural greedy algorithm provides a 2-competitive ratio.
Our main contributions for this problem are to
provide an efficient algorithm that is optimal for the 1-dimensional case,
improve the competitive ratio and prove a lower-bound for squares, and
%
provide an efficient implementation of the greedy algorithm for the general case.
We also prove a simple lower-bound for the general case.

To establish our results, we show interesting new connections between dispersion and ball-packing ---both {\em uniform} packing (i.e., with balls of identical radius) and {\em non-uniform} packing (i.e., with balls of different radii).
All our algorithms are deterministic and of polynomial time.
Some of them take an arbitrarily small constant~$\epsilon$
as a parameter and the running time is polynomial in $\frac{1}{\epsilon}$.
All inapproximability results hold even when running time is not a concern.
Most proofs are given in the appendix.

\begin{table}[htbp]
\begin{center}
\begin{tabular}{|c|p{0.45\textwidth}|p{0.35\textwidth}|}
  \hline
    & Online & Offline time-dependent \\ \hline
ATWC &
\begin{minipage}{0.45\textwidth}
\vspace{5pt}
$k=1$: a $2\ln 2$ ($\approx 1.386$) lower-bound and an optimal algorithm; see Theorems \ref{thm:1d-lower} and~\ref{thm:1-dtight}.

$k=2$: a $1.183$ lower-bound and a $1.591$-competitive algorithm for squares; see Theorems \ref{thm:2d-lower} and~\ref{thm:2d-worst}.

$k\geq 2$: a $\frac{7}{6}$ lower-bound and a $\frac{2}{1-\epsilon}$-competitive algorithm
for arbitrary polytopes~$P$;  see Theorems~\ref{thm:kd-lower} and \ref{thm:kd:approv}.
\vspace{5pt}
\end{minipage}
&
\begin{minipage}{0.35\textwidth}
Equivalent to
dispersion without  \\ time; see Claim \ref{clm:atwc=disp}.
\end{minipage} \\ \hline
CD & \begin{minipage}{0.45\textwidth}
No constant competitive algorithm even when $k=1$; see Claim~\ref{claim:ic:online}.
\end{minipage}
& \begin{minipage}{0.35\textwidth}
\vspace{5pt}
A black-box reduction to online ATWC for arbitrary $k$ and $P$, with the competitive ratio scaling up by at most 2; see Theorem \ref{thm:integral}.
\vspace{5pt}
\end{minipage}
\\ \hline
\end{tabular}
\end{center}
\caption{Online and offline time-dependent dispersion problems in a $k$-dimensional polytope $P$.}
\vspace{-15pt}
\label{table:summary}
\end{table}


\paragraph{Discussion and future directions.}
An algorithm for high-dimensional polytopes may not be directly applicable in dimension 1, because
all locations are on the boundary when a segment is treated as a high-dimensional polytope, and the minimum distance is always~0.
Accordingly, we do not know whether the lower-bound for dimension 1 carries through to higher dimensions, and proving better lower-bounds
will be an interesting problem for future studies. In the appendix, we consider online dispersion without the boundary constraint, where
the lower-bound for dimension~1 indeed carries through.
We show all our algorithms
can be adapted for this setting.
Another important future direction is to understand the role of randomized algorithms in the online dispersion problem. Finally, improving the (deterministic or randomized)
algorithms' competitive ratios in various classes of polytopes
is certainly a long-lasting theme for the online dispersion problem. Special classes such as regular polytopes and uniform polytopes may be reasonable starting points. Given the connections between dispersion and ball-packing,
it is conceivable that new competitive algorithms for online dispersion may stem from and also imply new findings on ball-packing.


%

\subsection{Related Work}

%

\paragraph{Dispersion without time.}
In dispersion problems in general,
the possible locations
can be
either a {\em continuous} region
or a set of {\em discrete} candidates.
Two objectives have been studied in the literature:
the {\em max-min distance} as considered in this paper,
 and the {\em maximum total distance.}
In continuous settings, the authors of \cite{baur2001approximation}
consider the max-min distance with the boundary condition.
Under $L_{\infty}$-norm, they give a polynomial-time 1.5-approximation in rectilinear polygons%
\footnote{A rectilinear polygon is a 2-dimensional polytope whose edges are axis-parallel.}
and show that a $\frac{14}{13}$-approximation in arbitrary polygons is NP-hard.
Moreover, they show there is no PTAS under any norm unless P=NP.
\cite{benkert2006polynomial} considers a similar boundary condition under $L_2$-norm
and provides a
$1.5$-approximation in polygons with obstacles.
\cite{dumitrescu2012dispersion} considers the problem of selecting $n$ points in
$n$ given unit-disks, one per disk, and the objective is to maximize the minimum distance.

In discrete settings,
\cite{ravi1994heuristic, birnbaum2009improved} show that,
if the distances among the candidate locations do not satisfy triangle inequality, then
there
is no polynomial-time constant approximation for either objective,
unless P=NP;
while if triangle inequality is satisfied, then there are efficient 2-approximations for
both objectives.
If the goal is to maximize total distance and the candidate locations are in a $k$-dimensional space,
\cite{fekete2004maximum} gives a PTAS under $L_{1}$-norm;
and \cite{cevallos2016max, cevallos2017local} provide PTASes
when locations need to satisfy matroid constraints.
Finally, \cite{ben2000obnoxious, qin2000some,katz2002improved, abravaya2010maximizing}
consider various dispersion problems in obnoxious facility allocation.

\paragraph{Packing without time.}
It is well known that dispersion and packing are ``dual'' problems of each other \cite{locatelli2002packing}.
In this paper we show interesting new connections between them
and use several important results for packing in our analysis.
Thus we briefly introduce this literature.
Indeed, the packing problem
is one of the most extensively studied problems
in geometric optimization, and a huge amount of work has been done on different variants of the problem; see \cite{peikert1992packing, hifi2009literature} for surveys on this topic.

One important problem is
to {\em pack circles with identical radius, as many as possible, in a bounded region.}
\cite{fowler1981optimal}
shows this problem to be
strongly NP-hard and \cite{hochbaum1985approximation} gives a PTAS for it.
An APTAS for the {\em circle bin packing} problem is given in \cite{miyazawa2014polynomial}.
The {\em dispersal packing} problem
tries to maximize the radius of a given number of circles packed in a square.
A lot of effort
has been made in finding the optimal radius and the corresponding packing when the number of circles
is a small constant;
see \cite{schaer1965geometric,schaer1965densest,schwartz1970separating,peikert1992packing}.
Heuristic methods
have also been used
in finding approximations when the number of circles gets large \cite{huang2010greedy,szabo2007packing}.
Finally, an important packing problem is to understand the {\em packing density}: that is, the maximum fraction of an infinite space covered by 
a packing of unit circles/spheres.
The packing density is solved for dimension~2
in \cite{fejes1942dichteste} and
for dimension 3 in \cite{hales2005proof}.
Very recently, \cite{viazovska2016sphere} and \cite{cohn2016sphere} solve it for dimensions~8 and 24, respectively.
Asymptotic lower bounds (as the dimension grows) for the density of the densest packing
are provided in \cite{rogers1947existence}.


\paragraph{Online geometric optimization.}
Many important geometric optimization problems have been studied in online settings, although
the settings
and the objectives
are quite different from ours.
In particular, the seminal work of \cite{seiden2002online} provides a nearly-optimal competitive algorithm
for the classic {\em online bin-packing} problem.
Algorithms for variants of the problem have been considered ever since,
such as a constant competitive ratio for packing circles in square bins \cite{hokama2016bounded},
and  constant competitive ratios for bin-packing in higher dimensions \cite{epstein2005optimal, epstein2007bounds}.

In {\em online facility location} \cite{meyerson2001online},
it is the demands rather than the facilities that come along time.
The facilities have open costs and the goal is to minimize
the total open cost and the total distance
between demands and facilities.
As shown in \cite{meyerson2001online},
when the demands arrive adversarially,
there is a randomized polynomial-time $O(\log n)$-competitive algorithm,
 and a constant competitive ratio is impossible.
A deterministic $O(\frac{\log n}{\log \log n})$-competitive algorithm and a matching lower-bound are provided in \cite{fotakis2008competitive} for the same problem.
In {\em incremental facility location} \cite{fotakis2006incremental},
the facilities can be opened, closed or merged, depending on the arriving demands.
In \cite{mettu2003online, rosenkrantz2006obtaining},
there is a cost for each location configuration
and the goal is to minimize the cost when the facilities arrive online.
A constant competitive algorithm for this problem is provided in \cite{mettu2003online},
and \cite{rosenkrantz2006obtaining} gives
a reduction from the online problem
to the offline version of the problem.

\paragraph{Dynamic resource division.}
Fair resource division is an important problem in economics \cite{ghodsi2011dominant,gutman2012fair,kash2014no}. When the resource is 1-dimensional and homogenous,
dynamic fair division is in some sense the ``dual'' of online dispersion: locating $n$ points as far as possible from each other
and from the boundary
is the same as partitioning the segment into $n+1$ pieces as evenly as possible.
\cite{friedman2015dynamic} provides
an optimal $d$-disruptive mechanism for 1-dimensional homogenous resource.
Interestingly, our algorithm for the 1-dimensional case provides an optimal
 mechanism when $d=1$,
although the techniques are
quite different.
Optimal mechanisms for heterogenous or high-dimensional resource remain unknown.
It would be interesting to see
 if our techniques for dispersion can be used in resource division problems in general.

\section{The Online Dispersion Problem}\label{sec:preliminary}

Given a $k$-dimensional polytope $P$, the {\em dispersion} problem \cite{baur2001approximation} takes as input a positive integer~$n$
and outputs $n$ locations, $X_1, \dots, X_n\in P$,
specifying how to locate $n$ points.
For each point $i$, let $dis(X_i, \partial P)$ be the distance from $X_i$ to $\partial P$, the boundary of $P$, measured by $L_2$-norm.
Also, let $dis(X_i, X_j)$ be the distance between $X_i$ and $X_j$ for any $i\neq j$.
The objective is
$$Disp(n; P) \triangleq \max_{X_1,\dots, X_n \in P}\ \min_{i, j\in [n]}\{dis(X_i, \partial P), dis(X_i, X_j)\}.$$
In Appendix~\ref{appendix:scattering} we also consider the dispersion problem where the distances to the boundary are not taken into consideration. Most of our techniques can be applied there.
%

We now define the {\em online dispersion} problem, where
each point $i$ arrives at time $s_i$ and departs at time $d_i$, with $d_i>s_i$.
Without loss of generality, $0=s_1\leq s_2 \leq \cdots\leq s_n$.
An online algorithm is notified upon the occurrence of  an arrival/departure event.
It must decide the location $X_i$ for a point $i$ upon its arrival, knowing neither the future events nor the total number of points $n$.
An adversary knows how the algorithm works and chooses future events after seeing the output of the algorithm so far.
In the {\em time-dependent offline} version of the problem, the times of all events,
denoted by a vector $S = ((s_1, d_1),\dots, (s_n, d_n))$, is given to the algorithm in advance.

Given such a vector $S$, let $T = \max_{i\in [n]} d_i$ be the last departure time.
Moreover, given locations $X = (X_1,\dots, X_n)$, for any $t\leq T$, let
$$d_{min}(t; X) = \min_{i, j\in [n] : s_i\leq t \leq d_i, s_j\leq t\leq d_i} \{dis(X_i, \partial P), dis(X_i, X_j)\}$$
be the minimum distance corresponding to the points that are present at time $t$.
When $X$ is clear from the context, we may write $d_{min}(t)$ for short.
We consider two natural objectives:
the {\em all-time worst-case} (ATWC) problem, where the objective is
$$OPT_A(S; P)  \triangleq \max_{X_1,\dots, X_n} \min_{t\leq T} d_{min}(t);$$
and the {\em cumulative distance} (CD) problem, where the objective is
$$OPT_C(S; P) \triangleq \max_{X_1, \dots, X_n} \int_0^T d_{min}(t) dt.$$
Note that both objectives
are defined to be the optimum of the corresponding {\em offline} problems,
the same
as the {\em ex-post} optimum for the online problems.
Below we provide two simple observations about the objectives, proved in Appendix \ref{apdx:preliminary}.

\begin{claim}
\label{claim:ic:online}
For the CD problem, even when $k=1$ and $P$ is the unit segment, no (randomized) online algorithm achieves a competitive ratio
to $OPT_C$ better than $\Omega(n)$.
\end{claim}

Next, given any ex-post instance $S = ((s_1, d_1),\dots, (s_n, d_n))$,
let $m$ be the maximum number of points simultaneously present at any time $t$:
that is, $m = \max_{t\leq T} |\{i : s_i\leq t\leq d_i\}|$.

\begin{claim}\label{clm:atwc=disp}
$\forall$ $S = ((s_1, d_1),\dots, (s_n, d_n))$, letting $m = \max_{t\leq T} |\{i : s_i\leq t\leq d_i\}|$, we have
$OPT_A(S; P) = Disp(m; P)$.
\end{claim}

In light of the claims above,
the online CD problem is highly inapproximable and the offline ATWC problem is equivalent to the dispersion problem without time.
Thus we will focus on the online ATWC problem and the offline CD problem, especially the former.
Our results will also imply a simple $O(n)$-competitive algorithm for the online CD problem (Corollary~\ref{col:onlineCD} of Appendix \ref{sec:kdatwc}), matching the lower-bound in Claim \ref{claim:ic:online}.

Below we point out some connections between dispersion and ball-packing, proved in Appendix \ref{apdx:preliminary}: they are not hard to show, and
similar results
for the dispersion problem without the boundary condition have been pointed out in \cite{locatelli2002packing}.
More precisely, the {\em (uniform) ball-packing} problem \cite{hochbaum1985approximation}
in a polytope $P$ takes as input a non-negative
value~$r$ and outputs an integer $n$,
the maximum number of balls of radius $r$
that can be packed non-overlappingly in~$P$, together with
a corresponding packing.
We denote the solution by $Pack(r; P)$.
The {\em dispersal packing} problem~\cite{baur2001approximation} is a ``mixture'' of dispersion and packing:
it takes as input an integer $n$ and
outputs the maximum radius for~$n$ balls with identical radius that can be packed in~$P$,
together with a corresponding packing.
That is,
$DP(n; P) \triangleq \max\{r:  Pack(r; P)\geq n\}$.

Recall that a $k$-dimensional convex polytope $P$ has an {\em insphere} if the largest ball contained wholly in $P$
is tangent to {\em all} the facets (i.e., $(k-1)$-faces) of $P$.
Such a ball, if it exists, is unique. It is referred to as {\em the} insphere of $P$.
The center of the insphere maximizes the minimum distance for any point in $P$ to its facets, and
has the same distance to all facets ---the radius of the insphere.
We have the following two claims.

\begin{claim}\label{clm:Disp-DP}
For any $k\geq 1$ and any $k$-dimensional convex polytope $P$ with an insphere,
letting~$x$ be the radius of the insphere,
we have
$Disp(n; P) = \frac{2xDP(n; P)}{x+DP(n; P)}$.
\end{claim}

\begin{claim}\label{clm:DispPacking}
For any $k\geq 1$ and any $k$-dimensional convex polytope $P$ with an insphere,
given the radius of the insphere,
\begin{itemize}
\item[(1)] any polynomial-time algorithm for $Disp(n; P)$ implies such an algorithm for $Pack(r; P)$;
\item[(2)] any polynomial-time algorithm for $Pack(r; P)$ implies an FPTAS for $Disp(n; P)$.
\end{itemize}
\end{claim}

To the best of our knowledge, it is still unknown whether ball-packing in regular polytopes
(which is a special case of convex polytopes with an insphere)
is NP-hard or not. Therefore the complexity of dispersion in regular polytopes remains open.
Note that ball-packing in arbitrary polytopes is NP-hard \cite{fowler1981optimal}, so is a $\frac{14}{13}$-approximation for
dispersion in rectilinear polygons \cite{baur2001approximation}.
Moreover, a claim similar to Claim \ref{clm:DispPacking} applies to $DP(n; P)$ and $Pack(r; P)$ in arbitrary polytopes.
The relation between dispersion and packing in arbitrary polytopes is not so clear and worth further investigation: for example,
it would be interesting to know if there exists a counterpart of Claim \ref{clm:DispPacking} when the polytope does not have an insphere.

%


Finally, the {\em insert-only} model,
where all points have the same departure time,
is a special case of our general model.
Interestingly, as will become clear in our analysis,
the difficulty of
the general online ATWC problem is captured by the
problem under this special model.
The insert-only model was also considered by \cite{mettu2003online, rosenkrantz2006obtaining}
in settings different from ours and with a different objective function.
We further discuss this model in Appendix \ref{app:insert_only}.


\section{The 1-Dimensional Online All-Time Worst-Case Problem}
\label{sec:1-D}

Note that a 1-dimensional polytope is simply a segment.
Without loss of generality, we consider the unit segment $P = [0, 1]$.
Below we first provide a lower bound for the competitive ratio of any algorithm, even computationally unbounded ones.
%

\subsection{The Lower Bound}


\begin{theorem}\label{thm:1d-lower}
No online algorithm
 achieves a competitive ratio better than $2\ln 2$~($\approx 1.386$) for the 1-dimensional ATWC problem.
\end{theorem}
\begin{proof}[Proof ideas]
Letting $\sigma'_r = \sum_{i=r+1}^{2r} \frac{1}{i}$ for
any positive integer $r$, we show that no algorithm achieves a competitive ratio better than $2\sigma'_r$.
Roughly speaking, we construct an instance (i.e., an adversary) for the online ATWC problem with three stages.
In the first stage, $r-1$ points arrive simultaneously;
in the second stage, $r$ new points arrive one by one;
and finally, all $2r-1$ points depart simultaneously.
If an algorithm $\cA$ is $\alpha$-competitive to $OPT_A$ with $\alpha<2\sigma'_r$,
it must be $\alpha$-competitive after the arrival of each point, as it does not know the
total number of points.
Thus for each arriving point, there must exist an interval long enough such that
putting the new point inside the interval does not violate the competitive ratio.
We show that in order for $\cA$ to do so, the segment must be longer than $P$ itself,
a contradiction.
Theorem \ref{thm:1d-lower} holds by setting $r \to \infty$.
The complete proof is provided in Appendix \ref{appendix:Theorem 1}.
\end{proof}

\subsection{A  Polynomial-Time Online Algorithm}

Next, we provide a deterministic polynomial-time online algorithm
whose competitive ratio to $OPT_A$ can be arbitrarily close to $2\ln 2$.
Intuitively, a good algorithm should disperse the points as evenly as possible.
However, if at some point of time with $m$ points present,
the resulting $m+1$ intervals on the segment have almost the same length,
then the next arriving point will force the minimum distance to drop by a factor of 2,
while the optimum only changes from $\frac{1}{m+1}$ to $\frac{1}{m+2}$,
causing the competitive ratio to drop by almost~2.
To overcome this problem,
the algorithm must find a balance between two consecutive points, choosing a sub-optimal solution for the former so as to leave enough space for the latter. The difficulty, as for online algorithms in general, is that this balance needs to be kept for arbitrarily many pairs of
consecutive point, as the sequence of points is chosen by an adversary who observes the algorithm's output.
Inspired by our lower bound,
roughly speaking,
our algorithm uses a parameter $r$
to pre-fix the locations of the  first $r$ points and the resulting $r+1$ intervals,
and then inserts the next $r+1$ points in the middle of these intervals.
The idea is that, when done properly, after these $2r+1$ points, the resulting configuration
is almost the same as if the algorithm has used parameter $2r+1$ to pre-fix the first $2r+2$ intervals:
then the procedure can repeat for arbitrary sequences.

More specifically, given a positive integer $r$,
let $Q = \{q_1, \dots, q_r\}$ be a set of positions on the segment,
such that the length ratios of the $r+1$ intervals sliced by them
are
$\frac{1}{r+1}: \frac{1}{r+2}: \cdots :\frac{1}{2r+1}.$
That is, letting $\sigma_r=\sum_{i=r+1}^{2r+1} \frac{1}{i}$, the lengths of the intervals are
$\frac{1}{\sigma_r(r+1)},\ \frac{1}{\sigma_r(r+2)},\ \cdots,\ \frac{1}{\sigma_r(2r+1)}$,
and $q_{i}=\frac{1}{\sigma_r}\sum_{j=r+1}^{r+i} \frac{1}{j}$ for each $i\in [r]$, as illustrated by Figure \ref{P-hat},
with $q_0 = 0$ and $q_{r+1} = 1$.
Note that $\sigma_r$ differs from $\sigma_r'$ in Theorem \ref{thm:1d-lower} by $\frac{1}{2r+1}$.
Also,
$\sigma_{r}$ is strictly decreasing in $r$ and
$\lim_{r\to\infty} \sigma_{r} = \ln 2$.
Moreover,
for any two intervals $(q_{j-1}, q_j)$ and $(q_{j'-1}, q_{j'})$ with $j< j'\leq r+1$,
  we have
$|(q_{j'-1}, q_{j'})|< |(q_{j-1}, q_j)| < 2|(q_{j'-1}, q_{j'})|$.

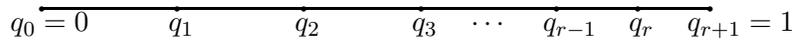
\begin{figure}[htbp]
\begin{center}
\setlength{\unitlength}{0.6cm}
\thicklines
\begin{picture}(15.8,1)

\put(0.5, 0.5){\line(1,0){14.8}}

\put(0.5, 0.5){\circle*{0.1}}
\put(-0.2, 0){$q_0 = 0$}

\put(3.5, 0.5){\circle*{0.1}}
\put(3.3, 0){$q_{1}$}

\put(6.3, 0.5){\circle*{0.1}}
\put(6.1, 0){$q_{2}$}

\put(8.9, 0.5){\circle*{0.1}}
\put(8.7, 0){$q_{3}$}

\put(10, 0){$\cdots$}

\put(11.9, 0.5){\circle*{0.1}}
\put(11.6, 0){$q_{r-1}$}

\put(13.7, 0.5){\circle*{0.1}}
\put(13.5, 0){$q_{r}$}

\put(15.3, 0.5){\circle*{0.1}}
\put(14.8, 0){$q_{r+1} = 1$}

\end{picture}
\caption{The pre-fixed positions in $Q$ for dimension 1.}

\label{P-hat}
\end{center}
\end{figure}
Our algorithm also takes as parameter an ordering for the positions in $Q$, denoted by $\tau = (\tau_1, \tau_2,\dots, \tau_r)$.
It is defined in Algorithm~\ref{alg:1d-online} and
we have the following two lemmas, proved in Appendix \ref{appendix:Lemma 2} and \ref{Appendix:claims567}, respectively.
We only sketch the main ideas below.
Recall that, given $S = ((s_1, d_1),\dots, (s_n, d_n))$,
$m$ is the maximum number of points simultaneously present at any time t.

\begin{algorithm}
  \caption{{\hspace{-3pt}{\bf .}} A polynomial-time algorithm for the 1-dimensional online ATWC problem.}
  \label{alg:1d-online}
  \begin{algorithmic}[1]
  \ENSURE A positive integer $r$, the corresponding set $Q=\{q_1,\dots, q_r\}$, and an ordering $\tau$ for $Q$.
  \REQUIRE  A sequence of points arriving and departing along time.
  \STATE Denote by $\hat{Q}$ the set of positions ever occupied by a point. At any point of time, a position in $\hat{Q}$ is labeled {\em occupied} if currently there is a point there and {\em vacant} otherwise. Initially $\hat{Q}=\emptyset$.
  \STATE When a point $i$ leaves, change the label of its position in $\hat{Q}$ from {\em occupied} to {\em vacant}.
  \STATE When a point $i$ arrives:
   \IF{$\hat{Q}=\emptyset$ or all positions in $\hat{Q}$ are labelled {\em occupied}}
    \IF{$Q\not\subseteq \hat{Q}$}
      \STATE\label{step:alg1-6} Choose the first position $q$ according to $\tau$ with $q\in Q\setminus \hat{Q}$,
       add it to $\hat{Q}$ and
       label it {\it occupied}.
      \STATE Put $i$ at position $q$.
    \ELSE
      \STATE Find position $q$ which is the middle of the largest interval created by the positions in $\hat{Q}$.
      \STATE Put $i$ at position $q$, add $q$ to $\hat{Q}$ and label it {\it occupied}.
    \ENDIF
  \ELSE
   \STATE Arbitrarily choose a vacant position $q$ from $\hat{Q}$ and label it {\it occupied}.
   \STATE Put $i$ at position $q$.
  \ENDIF
  \end{algorithmic}
\end{algorithm}

\begin{lemma}\label{lem:1d-online}
Given any $r$ and $\tau$,
Algorithm \ref{alg:1d-online} is $2\sigma_r$-competitive to $OPT_A$
for any $S$ with $m> r$.
\end{lemma}

\begin{proof}[Proof ideas]
Since Algorithm \ref{alg:1d-online} only creates a new position when the number of points
simultaneously present on the line increases,
for any time $t$, the number of positions created is exactly the maximum number of points that has appeared simultaneously on the line.
Thus only $m$ positions is created for instance $S$.
We prove that when $m > r$, the minimum distance
produced by our algorithm, denoted by
$d_{min}(m)$,
is $\frac{1}{2^{l+1}\sigma_r (r+i)}$,
where $l,i$ are the unique integers such that $l \geq 0, 0 \leq i \leq r+1$ and $2^{l}(r+1) + 2^{l}(i-1)\leq m < 2^{l}(r+1)+2^{l}i.$
Note that the minimum distance only depends on $m$.
By comparing $OPT_A$ with $d_{min}(m)$, we show that the competitive ratio $2\sigma_r$ holds for all $m > r$.
\end{proof}

%

%

\begin{lemma}\label{lem:1-dclose}
For any integer $l>0$ and $r = 2^l-1$, there exists an
ordering $\tau$ for the corresponding set $Q$,
s.t.
Algorithm \ref{alg:1d-online}
is $2\sigma_r$-competitive to $OPT_A$
for any
$S$
 with $m\leq r$.
\end{lemma}

\begin{proof}[Proof ideas]
Interestingly, due to the structure of Algorithm \ref{alg:1d-online},
we only need to consider the instance $S = ((1, r+1), (2, r+1), \dots, (r, r+1))$.
We construct an ordering $\tau = \{\tau_d\}_{d\in [r]}$ for $Q$ such that
the competitive ratio at any time $d\in [r]$  is smaller than $2\sigma_r$.
To do so, we fill in a complete binary tree with $r$ nodes as in Figure \ref{fig:tree},
and $\tau$ is obtained by traversing the tree in a breadth-first manner starting from the root.
Given any $d = 2^i + s$ with $i\in \{0, 1, \dots, l-1\}$ and $s\in\{0, 1, \dots, 2^i-1\}$,
we have
$\tau_d = q_{2^{l-i-1}(2s+1)}$.
Denoting the competitive ratio at time $d$ by $apx(d)$,
we prove that

$$apx(d) =\frac{\sigma_{r}}{(2^{i}+s+1)\cdot
\sum_{j=2^{l}+2^{l-i-1}(2s+1)}^{2^{l}-1+2^{l-i-1}(2s+2)}
\frac{1}{j}}.$$
Writing $apx(d)$ as $apx(i, s)$,
we prove that, fixing $i$, $apx(i, s)$ is strictly increasing in $s$;
and letting $s = 2^i - 1$, $apx(i, s)$ is strictly increasing in $i$.
Therefore the worst competitive ratio occurs at $i = l-1$ and $s = 2^{l-1} - 1$.
Since $apx(l-1, 2^{l-1} - 1) = (2- \frac{1}{2^l})\sigma_r < 2\sigma_r$, Lemma \ref{lem:1-dclose} holds.
\end{proof}

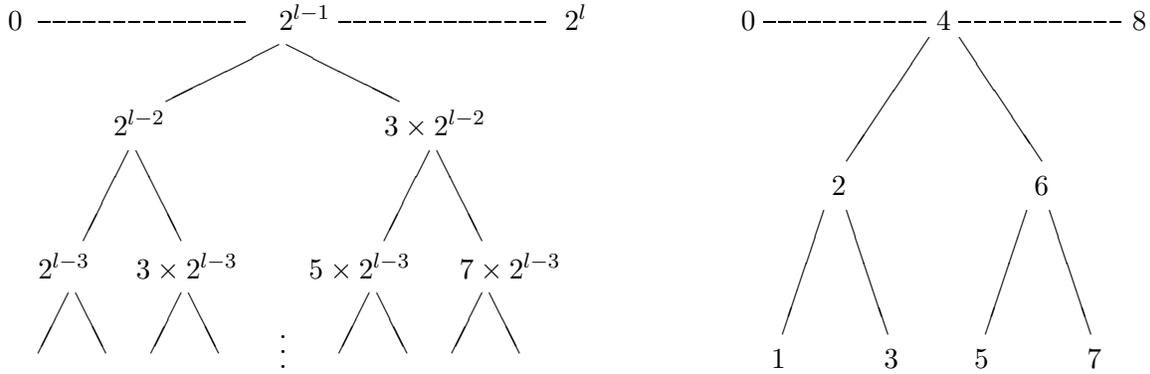
\begin{figure}[htbp]
\begin{flushleft}
\hspace{-5pt}
\begin{minipage}[t]{0.5\textwidth}
\centering
\setlength{\unitlength}{1cm}
\thinlines
\begin{picture}(8,4.5)

\put(0.4,4.5){$0$}
\put(7.8,4.5){$2^l$}
\put(4,4.5){$2^{l-1}$}

\multiput(0.8,4.6)(0.2,0){14}{\line(1,0){0.15}}
\multiput(4.8,4.6)(0.2,0){14}{\line(1,0){0.15}}

\put(4,4.3){\line(-2,-1){1.5}}
\put(4.1,4.3){\line(2,-1){1.5}}
\put(1.8,3.1){$2^{l-2}$}
\put(5.4,3.1){$3\times2^{l-2}$}

\put(2,2.9){\line(-1,-2){0.6}}
\put(2.1,2.9){\line(1,-2){0.6}}
\put(0.8,1.2){$2^{l-3}$}
\put(2.1,1.2){$3\times2^{l-3}$}

\put(6,2.9){\line(-1,-2){0.6}}
\put(6.1,2.9){\line(1,-2){0.6}}
\put(4.4,1.2){$5\times2^{l-3}$}
\put(6.4,1.2){$7\times2^{l-3}$}

\put(1.2,1){\line(-1,-2){0.4}}
\put(1.3,1){\line(1,-2){0.4}}
\put(2.7,1){\line(-1,-2){0.4}}
\put(2.8,1){\line(1,-2){0.4}}
\put(5.2,1){\line(-1,-2){0.4}}
\put(5.3,1){\line(1,-2){0.4}}
\put(6.7,1){\line(-1,-2){0.4}}
\put(6.8,1){\line(1,-2){0.4}}

\multiput(4,0)(0,.2){3}{.}

\end{picture}
\end{minipage}
\hspace{33pt}
\begin{minipage}[t]{0.35\textwidth}
\centering
\setlength{\unitlength}{1cm}
\thinlines
\begin{picture}(6,4.5)

\put(0.6,4.5){$0$}
\put(5.8,4.5){$8$}
\put(3.2,4.5){$4$}

\multiput(0.9,4.6)(0.2,0){11}{\line(1,0){0.15}}
\multiput(3.5,4.6)(0.2,0){11}{\line(1,0){0.15}}

\put(3.1,4.4){\line(-2,-3){1.1}}
\put(3.5,4.4){\line(2,-3){1.1}}
\put(1.8,2.3){$2$}
\put(4.5,2.3){$6$}

\put(1.7,2.1){\line(-1,-3){.55}}
\put(2,2.1){\line(1,-3){.55}}
\put(4.4,2.1){\line(-1,-3){.55}}
\put(4.7,2.1){\line(1,-3){.55}}

\put(1,0){$1$}
\put(2.5,0){$3$}

\put(3.7,0){$5$}
\put(5.2,0){$7$}

\end{picture}
\end{minipage}
\end{flushleft}

\caption{The left-hand side shows the top three levels of
the binary tree for a general $l$, with $\tau = (q_{2^{l-1}}, q_{2^{l-2}}, q_{3\times 2^{l-2}}, q_{2^{l-3}},
q_{3\times 2^{l-3}}, q_{5\times 2^{l-3}}, q_{7\times 2^{l-3}}, \dots)$.
The right-hand side shows the complete binary tree for
$l=3$, with $\tau=(q_{4},q_{2},q_{6},q_{1},q_{3},q_{5},q_{7})$.}
\label{fig:tree}
\end{figure}


The theorem below follows easily from the above two lemmas; see
Appendix \ref{Appendix:Theorem2}.

\begin{theorem}\label{thm:1-dtight}
There exists a deterministic polynomial-time online algorithm
for the ATWC problem, whose
competitive ratio can be arbitrarily close to
$2\ln 2$. Moreover, the running time is polynomial in $\frac{1}{\epsilon}$ for competitive ratio
$2\ln 2 + \epsilon$.
\end{theorem}

\paragraph{Remark.}
When the number of arrived points is large but
the maximum number $m$ of simultaneously present points is small,
the running time of the algorithm for each arriving point is polynomial in $m$
and can be much faster than being polynomial in the size of the input.


Following Theorem \ref{thm:1d-lower},
Algorithm \ref{alg:1d-online}
is essentially optimal.
Inspired by our constructions of $Q$ and $\tau$,
we actually characterize the optimal solution for the online ATWC problem, whose competitive ratio
is exactly $2\ln 2$: see Theorem \ref{thm:1-dexact} below, proved in Appendix \ref{Appendix:Theorem3}.
However, this solution involves irrational numbers and
cannot be exactly computed in polynomial time.

\begin{theorem}\label{thm:1-dexact}
For any integer $d=2^i+s$ with $i\geq 0$ and $0\leq s\leq 2^i-1$, let
$\tau_d=
\log_2 (1+\frac{2s+1}{2^{i+1}})$.
If Algorithm \ref{alg:1d-online}
creates
the $d$-th new position in $\hat{Q}$ to be $\tau_d$,
 the
competitive ratio is exactly $2\ln 2$.
\end{theorem}

\section{The 2-Dimensional Online All-Time Worst-Case Problem}
\label{sec:2D:ATWC}

We now
consider
 the 2-dimensional online ATWC problem in a square---without loss of generality, $P=[0, 1]^2$.
One difficulty is that,
different from the 1-dimensional problem where it is trivial to have
$Disp(n; P)= \frac{1}{n+1}$ for any $n\geq 1$,
here neither $Disp(n; P)$ nor $Pack(r; P)$ has
a known closed-form optimal
solution (whether polynomial-time computable or not). Accordingly,
 our lower-bound and our competitive
algorithm must rely on some proper upper- and lower-bounds for $Disp(n; P)$,
which is part of the reason why the resulting bounds are not tight.
In particular, we have the following lemma,
proved in Appendix \ref{appendix:clm:boundsdisp}.

\begin{lemma}\label{clm:boundsdisp}
For any $n\geq 1$, $\frac{2}{5+\sqrt{2\sqrt{3}n}} \leq Disp(n; P) \leq \frac{2}{2+\sqrt{2\sqrt{3}n}}$.
\end{lemma}

\subsection{The Lower Bound}

Interestingly, not only the dispersion problem
is closely related to {\em uniform} packing (i.e., the disks all have the same radius)
as we have seen in Section \ref{sec:preliminary},
but we also obtain a lower bound for the online ATWC problem
by carefully fitting a {\em non-uniform} packing
into the square.
The idea is to
imagine each position created
in an online algorithm
as a disk centered at that position.
The radius of each disk
is a function of the algorithm's competitive ratio
and
the optimal solutions to specific dispersion problems without time.
Note that the area covered by the disks is upper-bounded by the area of the square containing them.
Combining these relations together
gives us the following theorem, proved in Appendix \ref{appendix:thm:2d:lower}.

\begin{theorem} \label{thm:2d-lower}
No online algorithm achieves a competitive ratio better than 1.183 for the 2-dimensional ATWC problem in a square.
\end{theorem}

\subsection{A Polynomial-Time Online Algorithm}

\sloppy
Now we provide a deterministic polynomial-time online algorithm which is
1.591-competitive to $OPT_A$.
Similar to Algorithm \ref{alg:1d-online},
we first
construct
a set $Q$ of pre-fixed positions.
However, it is unclear
how to define $Q$ of arbitrary size in the square,
and we construct a set of 36 positions, denoted by $Q = \{q_1,\dots, q_{36}\}$.
It depends on a parameter $1< c <\sqrt{2}$ and $x=\frac{1}{3+4c}$,
as illustrated in
Figure~\ref{fig:36p:2nd}.
The indices of the $q_i$'s specify the order according to which they should be occupied,
thus we do not need an extra ordering $\tau$.
Note that
these positions create a grid in $P$
and split it
into multiple rectangles.
The choice of $c$ (and $x$, $Q$) will become clear in the analysis.

Whenever a new position needs to be created,
we pick the first position in $Q$
that has never been occupied
yet.
When all positions in $Q$ are occupied,
we may
(1) create a new position in the center of a current rectangle with the largest area,
split this rectangle into four smaller ones accordingly, and add the vertices of the new rectangles into the grid;
or (2) create a new position at a grid point that has never been occupied yet.
The main algorithm is
similar to Algorithm \ref{alg:1d-online} and
defined in Algorithm \ref{alg:2donline}.
It uses in Step
\ref{step:2-8}
a sub-routine, the {\em Position Creation Phase,}
as defined in Algorithm \ref{alg:2dcreation} in Appendix \ref{appendix:alg:2dcreation}.
In Appendix \ref{sec:intuition} we
provide some intuition on the choices of $Q$, $x$, and $c$.
By setting $c = 1.271$, we have the following theorem,
proved in Appendix \ref{App:thm:2d-worst}.

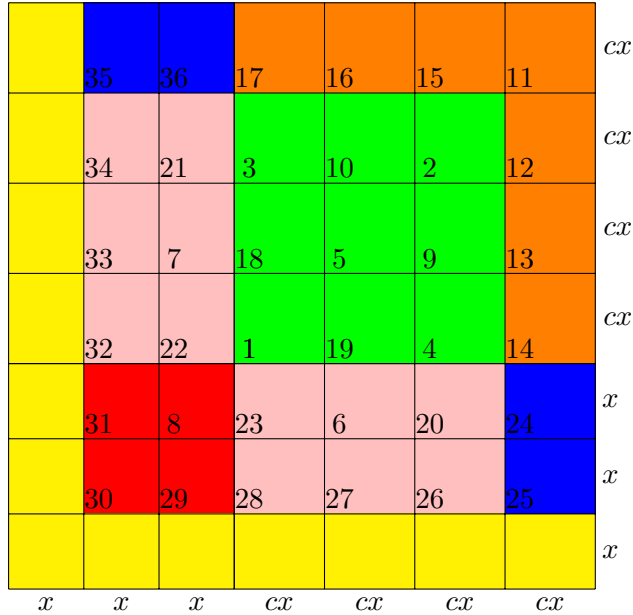
\begin{figure}[htbp]
\begin{center}
\begin{tikzpicture}

\fill [color=green](0,0) rectangle (3.6,3.6);
\fill [color=red](0,0) rectangle (-2,-2);
\fill [color=pink](0,0) rectangle (-2,3.6);
\fill [color=pink](0,0) rectangle (3.6,-2);
\fill [color=yellow](-3,-3) rectangle (-2,4.8);
\fill [color=yellow](-3,-3) rectangle (4.8,-2);
\fill [color=orange](3.6,0) rectangle (4.8,4.8);
\fill [color=orange](0,3.6) rectangle (4.8,4.8);
\fill [color=blue](-2,3.6) rectangle (0,4.8);
\fill [color=blue](3.6,0) rectangle (4.8,-2);

\draw[xstep=1cm,ystep=1] (0,0) grid (-3,-3);
\draw[xstep=1cm,ystep=1.2] (0,0) grid (-3,4.8);
\draw[xstep=1.2cm,ystep=1] (0,0) grid (4.8,-3);
\draw[xstep=1.2cm,ystep=1.2] (0,0) grid (4.8,4.8);

\path (0.2,0.2) node {$1$};
\path (2.6,2.6) node {$2$};
\path (0.2,2.6) node {$3$};
\path (2.6,0.2) node {$4$};
\path (1.4,1.4) node {$5$};
\path (0.2,0.2) node {$1$};
\path (1.4,-0.8) node {$6$};
\path (-0.8,1.4) node {$7$};
\path (-0.8,-0.8) node {$8$};
\path (2.6,1.4) node {$9$};
\path (1.4,2.6) node {$10$};
\path (3.8,3.8) node {$11$};
\path (3.8,2.6) node {$12$};
\path (3.8,1.4) node {$13$};
\path (3.8,0.2) node {$14$};
\path (2.6,3.8) node {$15$};
\path (1.4,3.8) node {$16$};
\path (0.2,3.8) node {$17$};
\path (0.2,1.4) node {$18$};
\path (1.4,0.2) node {$19$};
\path (2.6,-0.8) node {$20$};
\path (-0.8,2.6) node {$21$};
\path (-0.8,0.2) node {$22$};
\path (0.2,-0.8) node {$23$};
\path (3.8,-0.8) node {$24$};
\path (3.8,-1.8) node {$25$};
\path (2.6,-1.8) node {$26$};
\path (1.4,-1.8) node {$27$};
\path (0.2,-1.8) node {$28$};
\path (-0.8,-1.8) node {$29$};
\path (-1.8,-1.8) node {$30$};
\path (-1.8,-0.8) node {$31$};
\path (-1.8,0.2) node {$32$};
\path (-1.8,1.4) node {$33$};
\path (-1.8,2.6) node {$34$};
\path (-1.8,3.8) node {$35$};
\path (-0.8,3.8) node {$36$};

\path (-2.5,-3.2) node {$x$};
\path (-1.5,-3.2) node {$x$};
\path (-0.5,-3.2) node {$x$};
\path (0.6,-3.2) node {$cx$};
\path (1.8,-3.2) node {$cx$};
\path (3,-3.2) node {$cx$};
\path (4.2,-3.2) node {$cx$};

\path (5.1,4.2) node {$cx$};
\path (5.1,3) node {$cx$};
\path (5.1,1.8) node {$cx$};
\path (5.1,0.6) node {$cx$};
\path (5,-0.5) node {$x$};
\path (5,-1.5) node {$x$};
\path (5,-2.5) node {$x$};

\end{tikzpicture}

\caption{The set of pre-fixed positions, $Q = \{q_1,\dots, q_{36}\}$,
 and the grid created by $Q$.
A grid point labelled by $i$ indicates the position $q_i$.
The colored areas are used in the algorithm's description.
More specifically, denoting a rectangle by the position in $Q$ at its lower-left corner,
the green area is $(3, 10, 2; 18, 5, 9; 1, 19, 4)$;
the two pink areas are $(23, 6, 20; 28, 27, 26)$ and $(34, 21; 33, 7; 32, 22)$;
the red area is $(31, 8; 30, 29)$;
the two blue areas are $(35, 36)$ and $(24, 25)$;
the orange area is $(17, 16, 15, 11, 12, 13, 14)$;
and finally the yellow area contains all the remaining rectangles:
that is, rectangles adjacent to the left boundary
and the bottom boundary.}
\label{fig:36p:2nd}
\end{center}
\end{figure}

\begin{theorem} \label{thm:2d-worst}
Algorithm \ref{alg:2donline} runs in polynomial time and
is $1.591$-competitive for the 2-dimensional online ATWC problem in a square.
\end{theorem}

%


%

%

%


%

Note that the upper-bound for $Disp(n;P)$ in Lemma \ref{clm:boundsdisp} is not tight when $n$ is small.
With better upper-bounds for $Disp(n;P)$,
better competitive ratios for our algorithm can be
directly obtained via a similar analysis.
Moreover, we believe
the competitive ratio
can be improved by using
a larger set $Q$
and the best ordering for positions in $Q$.
Such a $Q$ and a rigorous analysis based on it
are
left for future studies.
Finally,
similar techniques can be used when $P$ is a rectangle,
but the gap between the lower- and upper-bounds will be even larger,
and the analysis will be more complicated without adding much new insight to the problem.
Thus we leave a thorough study on rectangles for the future.

\begin{algorithm}[hbtp]
  \caption{{\hspace{-3pt}{\bf .}} A polynomial-time online algorithm for the ATWC problem in a square.}
  \label{alg:2donline}
  \begin{algorithmic}[1]
  \ENSURE $c$ such that $1<c< \sqrt{2}$, the corresponding $x = \frac{1}{3+4c}$, and $Q$.
  \REQUIRE A sequence of points arriving and departing along time.
  \STATE Denote by $\hat{Q}$ the set of positions ever occupied by a point. At any point of time, a position in $\hat{Q}$ is labeled {\em occupied} if currently there is a point there and {\em vacant} otherwise. Initially $\hat{Q}=\emptyset$.
  \STATE When a point $w$ leaves,
  change the label of its position in $\hat{Q}$ from {\em occupied} to {\em vacant}.
  \STATE When a point $w$ arrives:
  \IF {$\hat{Q} = \emptyset$ or all positions in $\hat{Q}$ are labelled {\em occupied}}
  	\IF {$|\hat{Q}| < 36$}
  		\STATE Put $w$ at position $q_{|\hat{Q}|+1}$, add this position to $\hat{Q}$ and label it {\em occupied}.
	\ELSE
  		\STATE\label{step:2-8} Compute a position $q$ according to the Position Creation Phase defined
  in Algorithm~\ref{alg:2dcreation}.
        \STATE Put $w$ at position $q$, add $q$ to $\hat{Q}$ and label it {\em occupied}.
  	\ENDIF
  \ELSE
    \STATE Arbitrarily choose a vacant position $q$ from $\hat{Q}$ and label it {\em occupied}.
  	\STATE Put $w$ at position $q$.
  \ENDIF
  \end{algorithmic}
\end{algorithm}

\section{The General $k$-Dimensional Online ATWC Problem}
\label{sec:kdatwc}

Although the literature gives us little understanding about
the optimal dispersion/packing problem
in an arbitrary $k$-dimensional polytope $P$ with $k\geq 2$,
 we are still able to provide a simple lower-bound and a simple polynomial-time algorithm
for the online ATWC problem.
Below we only state the theorems; 
see Appendix~\ref{appendix:offline:ATWC} for the proofs.
%

\begin{theorem}
\label{thm:kd-lower}
For any $k\geq 2$,
no online algorithm achieves
a competitive ratio better than $\frac{7}{6}$ for the ATWC problem for arbitrary polytopes.
\end{theorem}


%


For any polytope $P$,
letting the {\em covering rate}
be the ratio between the edge-lengths of the maximum inscribed cube and the minimum bounding cube,
we have the following theorem.
Note that, although a natural greedy algorithm provides a 2-competitive ratio,
the exact greedy solution may not be computable in polynomial time.
Here we show the greedy algorithm
can be efficiently approximated arbitrarily closely.
The geometric problems
 of finding the minimum bounding cube, deciding whether a position is in $P$,
 and finding the distance between a point in $P$ and the boundary of $P$
 are given as oracles.

\begin{theorem}\label{thm:kd:approv}
For any constants $\gamma, \epsilon>0$, for any integer $k\geq 2$ and any $k$-dimensional polytope $P$ with covering rate at least $\gamma$,
there exists a deterministic polynomial-time online algorithm for the ATWC problem,
with competitive ratio $\frac{2}{1-\epsilon}$ and running time polynomial in $\frac{1}{(\gamma\epsilon)^k}$.
\end{theorem}

\section{The General $k$-Dimensional Offline CD Problem}
\label{sec:kdcd}

By Claim~\ref{claim:ic:online}, no online algorithm provides a good competitive ratio for the CD problem,
thus we focus on the offline problem.
%
Given an input sequence $S = ((s_1,d_1), \dots, (s_n, d_n))$,
we first slice the whole time interval $[0,T]$ into smaller ones
by the arriving time $s_{i}$ and the departure time $d_{i}$ of each point $i\in [n]$.
Thus the set of present points only changes at the end-points of the intervals
and stays the same within an interval.
Our algorithm will be such that,
{\em in each time interval,} the
minimum distance is a good approximation to the optimal dispersion problem
without time, for the points present in this interval.

Interestingly, this is achieved by reducing
the offline CD problem to the online ATWC problem,
 for any dimension~$k$
and polytope $P$.
To carry out this idea, we first provide a polynomial-time algorithm
$\cal A_I$ (Algorithm~\ref{FindSubset})
that, given a sequence $S$,
selects a subset $\cI$ of points from $S$.
The set $\cI$ satisfies the following properties,
which are proved in Claim \ref{claim:A_I} in Appendix \ref{appendix:offline:cd}.

\begin{description}
\item [$\Phi$.1] ${\cal I}$ can be partitioned into two groups $\cI_1$ and $\cI_2$ such that
the points in the same group have disjoint time intervals.
\item [$\Phi$.2] For any time $0\leq t \leq T$, if
there are points in $S$ present at time $t$,
 then at least one of them is selected to $\cal I$.
\end{description}


\begin{algorithm}[htbp]
  \caption{{\hspace{-3pt}{\bf .}} Algorithm $\mathcal{A_I}$ for computing $\cI = \{\cI_1, \cI_2\}$ satisfying properties $\Phi$.1 and $\Phi$.2.}
  \label{FindSubset}
  \begin{algorithmic}[1]
  \REQUIRE A sequence $S= ((s_1, d_1),\dots, (s_n, d_n))$.
  \STATE Let $\mathcal{I}_1= \mathcal{I}_2 = \emptyset$, $s=-1$, $d=0$ and  $T= \max_{i\in [n]} d_i$.

  ($s$ and $d$ are the end-points of a ``sliding window'' for the arriving times under consideration.)
  \STATE Let $index = 1$.
  \WHILE {d $\neq$ $T$}
  \STATE Let $\hat{S}=\{i| i \in S, s_i > s, s_i\leq d, d_i> d\}$.
  \IF {$\hat{S}\neq\emptyset$}
    \STATE\label{step:8-6} Arbitrarily choose $j \in \arg\max\limits_{i \in \hat{S}} d_i$ and add $j$ to
    $\mathcal{I}_{index}$.
    \STATE\label{step:8-8} $s=d$, then $d=d_j$.
  \ELSE
  \STATE\label{step:8-11} $s=d$, then $d=\min\limits_{i\in S, s_i>d} s_i$.
  \ENDIF
  \STATE $index = 3 - index$.
  \ENDWHILE
  \STATE Output $\cI = \{\cI_1, \cI_2\}$.
  \end{algorithmic}
\end{algorithm}

The offline CD algorithm $\cA_{CD}$ uses algorithm $\cal A_I$
to select $\cI$ from its input $S$, eliminates the selected points from $S$,
and repeats on the remaining $S$ until all points have been eliminated.
Recall that $m$ is the maximum number of points simultaneously present at any time.
By property $\Phi$.2,
this procedure ends in at most $m$ iterations.
Based on the partitions constructed by $\cal A_I$,
$\cA_{CD}$ constructs an instance of the online ATWC problem
and uses any online algorithm $\cA_{ATWC}$ for the latter as a black-box, so as to decide how to locate the
points.
Algorithm $\cA_{CD}$ is defined in Algorithm \ref{OFFline} and we have the following theorem, proved in Appendix \ref{appendix:offline:cd}.
Below we only sketch the main ideas.

\begin{algorithm}
  \caption{{\hspace{-3pt}{\bf .}} $\cA_{CD}$}
  \label{OFFline}
  \begin{algorithmic}[1]
  \REQUIRE A sequence $S= ((s_1, d_1),\dots, (s_n, d_n))$.
  \STATE Let $r = 0$.
  \WHILE {$S\neq \emptyset$}
    \STATE Run $\cal A_I$ on $S$ to obtain two disjoint sets $\mathcal{I}_{2r+1}, \mathcal{I}_{2r+2} \subseteq S$.
  \STATE $S =  S\setminus(\mathcal{I}_{2r+1}\cup\mathcal{I}_{2r+2})$.
    \STATE $r = r + 1$.
  \ENDWHILE
  \STATE\label{step:CD-7} Run $\mathcal{A}_{ATWC}$ on the following online sequence
  of $2r$ points:
     for all $i\in\{0,1,\dots,r-1\}$, points $2i+1$ and $2i+2$ arrive at time $i$. All points
     depart at time $r$.
  \STATE Letting $x_{2i+1}, x_{2i+2}$ be the two positions returned by $\cA_{ATWC}$ at time $i$,
  assign all points in $\mathcal{I}_{2i+1}$ to $x_{2i+1}$ and all points in $\mathcal{I}_{2i+2}$ to $x_{2i+2}$.
  \end{algorithmic}
\end{algorithm}

\begin{theorem}\label{thm:integral}
For any $k\geq 1$ and $k$-dimensional polytope $P$, given any polynomial-time online
algorithm $\mathcal{A}_{ATWC}$ for the ATWC problem with competitive ratio $\sigma$, there is a polynomial-time
 offline algorithm $\cA_{CD}$ for the CD problem with competitive ratio $\sigma\max\limits_{i\geq 1}\frac{Disp(i;P)}{Disp(2i;P)} \leq 2\sigma$, using $\mathcal{A}_{ATWC}$ as a black-box.
\end{theorem}

\begin{proof}[Proof ideas]
Given an input sequence $S$,
we slice the whole time interval $[0,T]$ into smaller ones
according to the arriving time and the departure time of each point.
Denote these small intervals by $T_1,\dots, T_l$, where $l$ is the number of small intervals created.
For each interval $T_i$,
let $S_{i}$ be the set of points that overlap with
$T_{i}$
and $n_{i}=|S_{i}|$.
By properties $\Phi$.1 and $\Phi$.2,
all points in $S_i$ are eliminated from $S$
in the first $n_i$ iterations of $\cA_{\cI}$,
thus are located at the first $2n_i$ positions
created by $\cA_{ATWC}$.
The minimum distance among points in $T_i$ (and to the boundary)
is at least $\frac{Disp(2n_{i};P)}{\sigma}$,
since algorithm $\cA_{ATWC}$ has competitive ratio $\sigma$.
Thus, within each $T_i$, the competitive ratio to the
optimal solution is upper-bounded by
$\frac{\sigma Disp(n_i;P)}{Disp(2n_i;P)}$.
Taking summation over all $T_i$'s, the competitive ratio is
upper-bounded by
$\sigma\max\limits_{i\geq 1}\frac{Disp(i;P)}{Disp(2i;P)}$.
Finally, we prove $\max\limits_{i\geq 1}\frac{Disp(i;P)}{Disp(2i;P)}\leq 2$,
finishing the proof of Theorem~\ref{thm:integral}.
\end{proof}

%


\paragraph{\bf Remark.}
Under the insert-only model,
it is not hard to see that
the online CD problem and the online ATWC problem are equivalent,
in the sense that an algorithm is $\sigma$-competitive for
 the online ATWC problem if and only if it is $\sigma$-competitive for the online CD problem.
Thus, all our algorithms for the online ATWC problem can be
directly applied to the online (and also offline) CD problem, with the competitive ratios unchanged.
Finally, all our inapproximability results for the online ATWC problem hold under the insert-only model.

\paragraph{Acknowledgements.}

The authors thank Joseph Mitchell
for motivating us to study the online dispersion problem.
We thank Esther Arkin, Michael Bender, Rezaul A. Chowdhury, Jie Gao, Joseph Mitchell, Jelani Nelson, and the participants of the Algorithm
Reading Group for helpful discussions, and several anonymous reviewers for helpful comments.

\bibliography{ref}
\bibliographystyle{abbrv}

\newpage

\appendix


%



%
%
%
%
%
%
%
%
%
%
%
%
%
%
%
%
%
%
%
%
%
%
%


\section{Proofs for Section \ref{sec:preliminary}}\label{apdx:preliminary}

\paragraph{Claim \ref{claim:ic:online}.} (restated) {\em
For the CD problem, even when $k=1$ and $P$ is the unit segment, no (randomized) online algorithm provides
a competitive ratio to $OPT_C$ better than $\Omega(n)$.
}


\begin{proof}
Consider $n$ points with $s_i = 0$ for all $i$
and let $X_1,\dots, X_n$ be the locations chosen by an online algorithm at time 0.
If there exist two points $i, j$ with $dis(X_i, X_j) = d_{min}(0)$, then
the adversary
sets the departure time of $i$ and $j$ to be a large number $T$,
and that of all the other points to be 1.
Otherwise, there exists a point $i$ with $dis(X_i, \partial P) = d_{min}(0)$, and
 the adversary sets $d_i =T$ and $d_{j} = 1$ for all $j\neq i$.
We only analyze the first case as the second is almost the same.
In the algorithm, $d_{min}(t) = d_{min}(0)\leq \frac{1}{n+1}$ for any $t$, and
$\int_0^T d_{min}(t) dt \leq \frac{T}{n+1}$.
However,
by putting $i$ and $j$ at $1/3$ and $2/3$ respectively, we have
$OPT_C(S; P)> \frac{T-1}{3}$. Thus the competitive ratio is $\Omega(n)$.
%
%
\end{proof}


\paragraph{Claim \ref{clm:atwc=disp}.} (restated) {\em
$\forall$ $S = ((s_1, d_1),\dots, (s_n, d_n))$, letting $m = \max_{t\leq T} |\{i : s_i\leq t\leq d_i\}|$, we have
$OPT_A(S; P) = Disp(m; P)$.
}

\begin{proof}
Let $X = (X_1,\dots, X_n)$ be the optimal solution
for the offline ATWC problem, and $t\in [0, T]$ be such that there exist exactly $m$ points at time $t$.
It is easy to see that
$$OPT_A(S; P) = \min_{t'\leq T} d_{min}(t'; X) \leq d_{min}(t; X) \leq Disp(m; P).$$
Next,
let $Y = (Y_1, \dots, Y_m)$ be the optimal solution for $Disp(m; P)$ and
consider the following algorithm for the offline ATWC problem with input $S$:
when a point arrives, arbitrarily pick a location $Y_i$ that is not currently occupied
and put it there; when a point leaves,
the~$Y_i$ occupied by it becomes vacant again. Note that this is an offline
algorithm because it knows~$m$ (and thus~$Y$). Also note that this algorithm produces a valid solution,
because there are at most $m$ points simultaneously present at any time, and $m$ locations are sufficient.
Abusing notation sightly and letting $d_{min}(t; Y)$ be the minimum distance produced by the algorithm at time $t$, we have
$d_{min}(t; Y)\geq Disp(m; P)$ for all $t\leq T$, thus
$$Disp(m; P) \leq OPT_A(S; P).$$
Therefore Claim \ref{clm:atwc=disp} holds.
\end{proof}


\paragraph{Claim \ref{clm:Disp-DP}.} (restated) {\em
For any $k\geq 1$, $x>0$ and any $k$-dimensional polytope $P$
with an insphere,
letting $x$ be the radius of the insphere,
we have
$Disp(n; P) = \frac{2xDP(n; P)}{x+DP(n; P)}$.
}

\begin{proof}
Let $c$ be the center of the insphere.
On the one hand,
given the optimal locations $X_1,\dots, X_n$ corresponding to $Disp(n; P)$,
let $P'$ be the polytope obtained from $P$ by moving each facet towards~$c$ for distance $\frac{Disp(n; P)}{2}$.
As the distance from $c$ to each facet of $P$ is exactly $x$,
$P'$ can also be obtained by
shrinking $P$ by a factor of $\lambda = \frac{x}{x-\frac{Disp(n; P)}{2}}$
with respect to the origin at $c$.
Consider the packing of $n$ balls with radius $\frac{Disp(n; P)}{2}$ in $P'$,
centered at the $X_i$'s.
Indeed,
as the distance of each $X_i$ to the facets of $P$ is at least $Disp(n; P)$,
its distance to the facets of $P'$ is at least $\frac{Disp(n;P)}{2}$
and the $n$ balls are contained wholly in $P'$.
Moreover,
as the distance of any two locations $X_i$ and $X_j$ is at least $Disp(n; P)$,
the $n$ balls are not overlapping with each other.
By scaling $P'$ up by a factor of $\lambda$ with respect to $c$,
we get a packing of $n$ balls in $P$ with radius $r = \frac{\lambda Disp(n; P)}{2} = \frac{x Disp(n; P)}{2x-Disp(n; P)}$.
Thus $Pack(r; P)\geq n$. Accordingly, $DP(n; P)\geq r$ by definition, which implies
$$Disp(n; P) \leq \frac{2x DP(n; P)}{x + DP(n; P)}.$$

On the other hand,
given the optimal solution for $DP(n; P)$, with the balls centered at $Y_1,\dots, Y_n$, let $P''$ be the polytope obtained from $P$ by moving each facet away from  $c$
for distance
$DP(n; P)$.
It is easy to see that $Y_1,\dots, Y_n$ is a dispersion in $P''$ with distance $2DP(n; P)$.
Again because the distance from $c$ to each facet of $P$ is exactly $x$, $P''$ can be obtained by scaling $P$ up
by a factor of $\lambda' = \frac{x+DP(n; P)}{x}$ with respect to $c$.
By scaling $P''$ down by a factor of $\lambda'$, we obtain a dispersion in~$P$ with distance
$d = \frac{2DP(n; P)}{\lambda'} = \frac{2x DP(n; P)}{x+DP(n; P)}$. Therefore
$$Disp(n; P) \geq \frac{2x DP(n; P)}{x + DP(n; P)},$$
and Claim \ref{clm:Disp-DP} holds.
\end{proof}

\paragraph{Claim \ref{clm:DispPacking}.} (restated) {\em
For any $k\geq 1$ and any $k$-dimensional convex polytope $P$ with an insphere,
given the radius of the insphere,
\begin{itemize}
\item[(1)] any polynomial-time algorithm for $Disp(n; P)$ implies such an algorithm for $Pack(r; P)$;
\item[(2)] any polynomial-time algorithm for $Pack(r; P)$ implies an FPTAS for $Disp(n; P)$.
\end{itemize}
}

%
\begin{proof}
Roughly speaking, given an algorithm for one of the two problems,
 we can use binary search to find a solution for the other.
Without loss of generality, assume $P$ has volume 1.

Let $x$ be the radius of the insphere.
For the first part of the claim,
let $Ball(r)$ be the volume of the $k$-dimensional ball with radius $r$ and $N = \lfloor \frac{1}{Ball(r)} \rfloor$.
Clearly, $Pack(r; P)\leq N$. The binary search over the set $\{0, 1, \dots, N\}$ works as follows.
In each round, find the median of the current set, denoted by $n$;
compute $d_n = Disp(n; P)$ using the polynomial-time algorithm and $r_n = \frac{d_n x}{2x-d_n}$. Note $r_n=DP(n; P)$ by Claim \ref{clm:Disp-DP}.
If $r_n < r$ then continue searching from $n-1$ and below---that is, one cannot pack $n$ balls of radius $r$ in $P$; and
if $r_n\geq r$ then continue searching from $n$ and above---that is, one can pack at least $n$ balls of radius $r$ in $P$.
When $n$ is the only number left, output it.
The correctness of the algorithm follows from the fact that $Disp(n; P)$ and thus $DP(n; P)$ are non-increasing in $n$.
Accordingly, there exists a unique $n^*$
such that $r_n\geq r$ for all $n\leq n^*$ and $r_{n}<r$ for all $n>n^*$. It is easy to see that $Pack(r; P) = n^*$.

For the second part, using $x$ as an upper bound for $Disp(n; P)$, the idea is almost the same.
The only difference is that $Disp(n; P)$ may not have finite representations and the algorithm has to stop when the length of the interval
is no larger than some small constant~$\epsilon$, leading to an FPTAS rather than an exact solution.
\end{proof}

Note that finding the radius of the insphere in an irregular convex polytope is a non-trivial computation problem.
 Thus we require it is given as an input. Finding the radius is easy for regular polytopes.

\section{Proofs for Section \ref{sec:1-D}}
\label{apdx:1-D}

\subsection{Proof for Theorem \ref{thm:1d-lower}}
\label{appendix:Theorem 1}

\paragraph{Theorem \ref{thm:1d-lower}.} (restated) {\em
No algorithm
 achieves a competitive ratio better than $2\ln 2$~($\approx 1.386$) for the 1-dimensional online ATWC problem.
}

\begin{proof}
For any positive integer $r$, let $\sigma'_r = \sum_{i=r+1}^{2r} \frac{1}{i}$.
We show that no algorithm achieves a competitive ratio better than $2\sigma_r'$.
For the sake of contradiction, assume there exists an $r$ and
an online algorithm $\mathcal{A}$
with competitive ratio $\alpha< 2\sigma_r'$.
We construct an instance of the online ATWC problem with three stages.
In the first stage, $r-1$ points arrive simultaneously;
in the second stage, $r$ new points arrive one by one;
and then all $2r-1$ points depart simultaneously.

Since no point departs before the last point arrives,
by Claim \ref{clm:atwc=disp} we have $OPT_A(S_i; P) = Disp(i; P)$ for any point $i$, where $S_i$ is the instance containing
only the first $i$ points.
Since the online algorithm
is $\alpha$-competitive to $OPT_A$,
it must ensure that for each point $i$, after its arrival, its distance to all the other present points is at least $\frac{OPT_A(S_i; P)}{\alpha} = \frac{Disp(i; P)}{\alpha}$:
otherwise the adversary simply stops adding new points and
the competitive ratio is violated for the instance~$S_i$.

Nevertheless, we show that after the second stage,
the claimed competitive ratio must be violated.
To do so,
note that $Disp(i; P) = \frac{1}{i+1}$ for any point $i$, because the optimal dispersion without time
is to locate the points evenly on the segment, resulting in $i+1$ equal-length intervals.
Denote by $Q = \{q_\ell | 1\leq \ell \leq r-1\}$ the positions
of the first $r-1$ points given by the algorithm, and let $q_0 = 0$ and $q_r = 1$.
We claim that, in the second stage, no two points can be put into the same interval generated by~$Q$.
Assume otherwise, and
assume exactly two points $i$ and $j$ are put into the same interval, with $r\leq i < j \leq 2r-1$.
There must exist another interval $(q_{\ell},q_{\ell+1})$ with $\ell\geq 0$,
which does not contain any new point from stage two.
Any other interval contains exactly one new point.
Thus the total length of all the
 intervals generated by $Q$, denoted by $\mathcal{L}$, is
\begin{equation}\label{equ:2}
\mathcal{L} > (q_{\ell+1}-q_{\ell}) + \dfrac{Disp(i;P)}{2\sigma_r'} + 2\times\dfrac{Disp(j;P)}{2\sigma_r'} + \sum\limits_{r\leq h \leq 2r-1,\ h\not\in\{i,j\}} 2\times \dfrac{Disp(h;P)}{2\sigma_r'}.
\end{equation}

To see why Equation \ref{equ:2} is true, first note that $q_{\ell+1}-q_{\ell}$ is the
length of the interval which does not contain any new point.
Second, the length of the interval
split by $i$ and $j$ must be larger than
$$\frac{Disp(i;P)}{2\sigma_r'} + 2\times\frac{Disp(j;P)}{2\sigma_r'}.$$
Indeed,
upon arrival, point $i$ first splits this interval into two smaller intervals,
and one of them is further split by $j$; thus the sub-interval between $i$ and the adjacent end-point of the interval
is at least $\frac{Disp(i; P)}{\alpha}> \frac{Disp(i; P)}{2\sigma_r'}$,
and each of the two sub-intervals created by $j$ is at least
$\frac{Disp(j; P)}{\alpha}> \frac{Disp(j; P)}{2\sigma_r'}$,
due to the algorithm's claimed competitive ratio.
Moreover, for each point $h\in \{r, r+1,\dots, 2r-1\}\setminus\{i,j\}$, it splits one of the remaining intervals and
each of the two sub-intervals is at least $\frac{Disp(h; P)}{\alpha} > \frac{Disp(h; P)}{2\sigma_r'}$.

Again because algorithm  ${\cal A}$'s competitive ratio is $\alpha< 2\sigma_r'$, we have
\begin{equation}\label{equ:3}
q_{\ell+1}-q_{\ell} \geq \frac{Disp(r-1; P)}{\alpha} = \frac{1}{\alpha r}> \dfrac{1}{2\sigma_r' r} > \dfrac{1}{2\sigma_r' (i+1)} = \dfrac{Disp(i;P)}{2\sigma_r'}.
\end{equation}
Combining Equations \ref{equ:2} and \ref{equ:3}, we have
\begin{equation}\label{equ:3_1}
\mathcal{L} > \sum\limits_{r\leq h \leq 2r-1} 2\times \dfrac{Disp(h;P)}{2\sigma_r'} = \dfrac{1}{\sigma_r'} \sum\limits_{r\leq h \leq 2r-1} \dfrac{1}{h+1} = \dfrac{1}{\sigma_r'} \sum\limits_{r+1\leq h \leq 2r} \dfrac{1}{h} = 1.
\end{equation}
That is, the total length is larger than 1, a contradiction.

In general, having more than two points in the same interval
and having more than one interval
containing at least two points lead to the same contradiction.
More specifically, for all new points $i$ that contributes one copy of $\frac{Disp(i; P)}{2\sigma'_r}$ in the lower bound of $\cL$ as in Equation~\ref{equ:2},
there exists exactly the same number of intervals that are not split by any new point.
Thus we can arbitrarily fix a bijection between those points and intervals,
such that each unsplit interval contributes another copy of
$\frac{Disp(i; P)}{2\sigma'_r}$ for its corresponding point $i$ as in Equation \ref{equ:3},
and Equation \ref{equ:3_1} holds again.
Accordingly, we conclude that each interval generated by $Q$ is split by exactly one new point in stage two.

However, in this case, by a similar argument,
the total length of these intervals is
$$\mathcal{L} > \sum\limits_{r\leq h \leq 2r-1} 2\times \dfrac{Disp(h; P)}{2\sigma_r'} = \dfrac{1}{\sigma_r'} \sum\limits_{r\leq h \leq 2r-1} \dfrac{1}{h+1} = \dfrac{1}{\sigma_r'} \sum\limits_{r+1\leq h \leq 2r} \dfrac{1}{h} = 1,$$
the same contradiction.
Therefore such an algorithm $\cA$ does not exist.

In sum, for any integer $r>0$, no  algorithm achieves a competitive ratio better than $2\sigma_r'$.
Since $\sigma_r'$ is strictly increasing in $r$ and $\lim_{r\to \infty} 2\sigma_r' = 2\ln2$,
no algorithm achieves a competitive ratio better than $2\ln 2$ and Theorem \ref{thm:1d-lower} holds.
\end{proof}

\subsection{Proof for Lemma \ref{lem:1d-online}}
\label{appendix:Lemma 2}

\paragraph{Lemma \ref{lem:1d-online}.} (restated) {\em
%
Given any $r$ and $\tau$,
Algorithm \ref{alg:1d-online} is $2\sigma_r$-competitive to $OPT_A$
for any
$S= ((s_1, d_1),\dots, (s_n, d_n))$
with $m> r$, where $m$ is maximum number of points simultaneously present at some time $t$.
}

\begin{proof}
Note that for each arriving point Algorithm \ref{alg:1d-online} first checks the existing
positions in $\hat{Q}$.
Only when there is no {\it vacant} position will it create a new one, with the pre-fixed positions in $Q$ having the highest priority.
In particular, whenever a new position is added to $\hat{Q}$ at time $\hat{t}$,
$\min_{t<\hat{t}} d_{min}(t) = d_{min}(\hat{Q})$,
where $d_{min}(\hat{Q})\triangleq \min_{q, q'\in \hat{Q}} \{dis(q, \partial P), dis(q, q')\}$ is the minimum distance incurred by positions in $\hat{Q}$.
Accordingly, the minimum distance throughout Algorithm \ref{alg:1d-online}
happens when there are $m$ points on the line.
Note $OPT_A(S; P) = Disp(m; P) = \frac{1}{m+1}$ by Claim \ref{clm:atwc=disp}.
Thus, in order to lower bound the competitive ratio of the algorithm,
it suffices to consider the case
where each point $i$ arrives at time $i$ and all points depart at the same time $m+1$---
that is,
$S=((1,m+1),(2,m+1),\dots,(m,m+1))$. By the hypothesis of Lemma \ref{lem:1d-online}, $m> r$.


By the construction of the algorithm, upto time $r$, only the positions in $Q$ may be used for the points.
At time $r$, $\hat{Q} = Q$, all $r$ positions in $Q$ are {\it occupied} and the segment is sliced into $r+1$ intervals.
After that, at time $i>r$, the arriving point $i$ will split the current largest interval into
 two equal sub-intervals and
create a new position in $\hat{Q}$.
In fact, since the lengths of the intervals created by $Q$ strictly decrease from left to right,
the position creation procedure can be described by ``rounds'' as follows.
Round 0 is
$\tau_1, \dots, \tau_r$
(from time 1 to time $r$);
round $1$ splits existing intervals one by one into halves,
from the leftmost to the rightmost (from time $r+1$ to time ${2(r+1)}$);
round $2$ again splits existing intervals
from the leftmost to the rightmost (from time ${2(r+1)+1}$ to time ${4(r+1)}$);
 so on and so forth.
In particular,
by the end of each round $l$, the number of sub-intervals sliced by $\hat{Q}$ is $2^{l}(r+1)$
and the number of points
is $|\hat{Q}| = 2^l(r+1)-1$.
Also, for each interval $(q_{i-1}, q_i)$ with $i\in [r+1]$,
it has been sliced
into $2^l$ sub-intervals after round $l$.
Moreover, because the maximum interval $(q_0, q_1)$ is less than twice of the minimum interval $(q_r, q_{r+1})$,
whenever a point $i>r$ is added and $i\in (q_{j-1}, q_j)$ for some $j$,
the resulting minimum distance occurs between $i$ and its right neighbor in $\hat{Q}$
(and is also the length of all sub-intervals in $(q_{j-1}, q_j)$ before $i$).

Let point $m$ appear during round $(l+1)$ for some $l\geq 0$, then
there exits $1\leq i \leq r+1$ such that
$$2^{l}(r+1) + 2^{l}(i-1)\leq m < 2^{l}(r+1)+2^{l}i.$$
That is, $m$ appears in the interval $(q_{i-1}, q_{i})$.
Accordingly, after point $m$ is located, the minimum distance equals
the length of every sub-interval in $(q_{i-1}, q_i)$ before point $m$, as well as the sub-interval immediately after it.
And this length is
$$
d_{min}(m) = \frac{1}{2^{l+1}\sigma_r (r+i)}.
$$
Therefore the competitive ratio is
$$\frac{OPT_A(S; P)}{d_{min}(m)} = \dfrac{\dfrac{1}{m+1}}{\dfrac{1}{2^{l+1}\sigma_r (r+i)}} = \dfrac{2^{l+1}\sigma_r (r+i)}{m+1}$$
$$\leq \dfrac{2^{l+1}\sigma_r (r+i)}{2^{l}(r+1) + 2^{l}(i-1)+1}
=  \frac{2^{l+1}\sigma_r (r+i)}{2^l(r+i)+1} < 2\sigma_r,$$
and Lemma \ref{lem:1d-online} holds.
\end{proof}

\subsection{Proof for Lemma \ref{lem:1-dclose}}
\label{Appendix:claims567}

\paragraph{Lemma \ref{lem:1-dclose}.} (restated) {\em
\sloppy
%
For any positive integer $l$ and $r = 2^l-1$, there exists an
ordering $\tau$ for the corresponding set~$Q$,
such that
Algorithm \ref{alg:1d-online}
is $2\sigma_r$-competitive to $OPT_A$
for any
$S= ((s_1, d_1),\dots, (s_n, d_n))$
 with $m\leq r$.
}

\begin{proof}
As shown in Algorithm \ref{alg:1d-online}, for each arriving point,
when $m\leq r$,
it always assigns the point to the existing positions in $\hat{Q}$, and when there is no vacant position,
it creates a new one at the $(|\hat{Q}|+1)$st position in $\tau$.
Thus, it is again
sufficient to consider all the instances
$S = ((1, m+1), (2, m+1),\dots, (m, m+1))$ with $m\leq r$:
whenever a new position is created,
the minimum distance so far is incurred by $\hat{Q}$, and the size of $\hat{Q}$ only increases.
In other words, we can focus on the instance $S = ((1, r+1), (2, r+1), \dots, (r, r+1))$
and prove that the competitive ratio at any time $d\in [r]$ is smaller than $2\sigma_r$.

Below we construct the desired ordering $\tau$ for $Q$ by filling in a complete binary tree with~$r$ nodes.
We do so in $l$ rounds,
with round $j\in \{0, \dots, l-1\}$ filling in level $j$ of the tree and level 0 being the root.
In round 0, letting the {\em left end-point} be $0$ and the {\em right end-point} be $2^l$ (corresponding to $q_0 = 0$ and $q_{2^l} = 1$),
fill the root with the average of the two, namely, $2^{l-1}$.
In each round $j>0$, process the nodes in level $j$ from left to right.
For each node $x$, letting its two neighbors in the current tree be filled with $x_{left}$ and $x_{right}$,
fill node $x$
with $\frac{x_{left} + x_{right}}{2}$.
If node $x$ is the leftmost (respectively, rightmost) node in the current tree, then take $x_{left} = 0$ (respectively, $x_{right} = 2^l$).
After the whole tree being filled, $\tau$ is obtained by traversing it in a breadth-first manner starting from the root: letting the $j$-th node visited being filled with $x_j$,
then $\tau_j = q_{x_j}$.
Figure \ref{fig:tree} illustrates the structure of the tree and the ordering $\tau$, for general $l$ and for $l=3$.
We refer to the resulting~$\tau$ as the {\em binary ordering} and we have the following claim.

\begin{claim}\label{clm:1-d:formula}
\sloppy
%
%
For any positive integer $l$ and
$r = 2^l-1$,
for any $d\in [r]$, writing $d = 2^i + s$ with $i\in \{0, 1, \dots, l-1\}$ and $s\in\{0, 1, \dots, 2^i-1\}$,
then the binary ordering $\tau$ is such that
$$\tau_d = q_{2^{l-i-1}(2s+1)},$$
and the minimum distance according to Algorithm \ref{alg:1d-online} with parameter $r$ and $\tau$ at time $d$, $d_{min}(d)$, is
$$d_{min}(d) = \frac{1}{\sigma_r}\sum_{j=2^l+2^{l-i-1}(2s+1)}^{2^l-1+2^{l-i-1}(2s+2)} \frac{1}{j}.$$
\end{claim}

We prove Claim \ref{clm:1-d:formula} after
the proof of
Lemma \ref{lem:1-dclose}.
Note that $OPT_A(S_d;P) = Disp(d; P) = \frac{1}{d+1} = \frac{1}{2^i+s+1}$ for $d=2^i+s$, where $S_d$ contains only the first $d$ points in $S$.
Thus, by Claim~\ref{clm:1-d:formula} the competitive ratio at time $d$, $apx(d)$, is
$$apx(d) = \frac{OPT_A(S_d;P)}{d_{min}(d)}=\frac{\sigma_{r}}{(2^{i}+s+1)\cdot
\sum_{j=2^{l}+2^{l-i-1}(2s+1)}^{2^{l}-1+2^{l-i-1}(2s+2)}
\frac{1}{j}}.$$
To show $apx(d)< 2\sigma_r$,
we lower bound the denominator in two steps, by the following two claims, which are also proved after the proof of
Lemma \ref{lem:1-dclose}.

\begin{claim}\label{clm:f(s)}
%
Arbitrarily fixing $i\in \{0, 1, \dots, l-1\}$ and
letting
$$f(s)=(2^{i}+s+1)\cdot
\sum\limits_{j=2^{l}+2^{l-i-1}(2s+1)}^{2^{l}-1+2^{l-i-1}(2s+2)}
\frac{1}{j}$$
for all $s \in \{0, 1, \dots, 2^{i}-1\}$,
we have that $f(s)$ is strictly decreasing in $s$.
\end{claim}

By Claim \ref{clm:f(s)},
$apx(d) = \frac{\sigma_r}{f(s; i)} \leq \frac{\sigma_r}{f(2^i-1; i)} = \frac{\sigma_{r}}{(2^{i}+2^{i})\cdot
\sum_{j=2^{l}+2^{l-i-1}(2^{i+1}-1)}^{2^{l}-1+2^{l-i-1}2^{i+1}}\frac{1}{j}}
=\frac{\sigma_{r}}{2^{i+1}\cdot
\sum_{j=2^{l+1}-2^{l-i-1}}^{2^{l+1}-1}\frac{1}{j}}$.

\begin{claim}\label{clm:g(i)}

Letting
$$
g(i)=2^{i+1}\cdot
\sum\limits_{j=2^{l+1}-2^{l-i-1}}^{2^{l+1}-1}\frac{1}{j}
$$
for all $i\in\{0,1,\cdots,l-1\}$,
we have that $g(i)$ is strictly decreasing in $i$.
\end{claim}

By Claim \ref{clm:g(i)},
for any $d\in [r]$,
$
apx(d)\leq
\frac{\sigma_{r}}{g(l-1)}=
\frac{\sigma_{r}}{2^{l}\cdot \sum_{j=2^{l+1}-1}^{2^{l+1}-1}\frac{1}{j}}
=\frac{\sigma_{r}}{2^{l}/(2^{l+1}-1)}
= \sigma_{r}(2-\frac{1}{2^{l}})< 2\sigma_{r},
$
and Lemma \ref{lem:1-dclose} holds.
\end{proof}

\paragraph{Remark.} Note the denominator of $apx(d)$ may not be
 decreasing in $d$,
and the jump may happen from $d= 2^i+(2^i-1)$ to $2^{i+1}$.
We bypass this problem by breaking the analysis into two steps, as above.

\medskip

We now prove the three claims.

%

\begin{proof}[Proof of Claim \ref{clm:1-d:formula}]
%
%
%
We first provide some simple facts about $\tau$ and Algorithm \ref{alg:1d-online}.
Since the algorithm
adds positions to $\hat{Q}$ according to $\tau$, the whole procedure can also be considered as $l$ rounds,
same as the construction of the binary tree.
By induction, for each $i = 0, 1, \dots, l-1$,
the number of positions added to $\hat{Q}$ in round $i$ is $2^i$
(i.e., the number of nodes in level~$i$ of the tree),
and
the number of intervals created by $\hat{Q}$ by the end of round $i$ is
$2^{i+1}$.
We denote these intervals by $I^i_0,I^i_1,\dots, I^i_{2^{i+1}-1}$ from left to right, and refer to them as the {\em round-$i$ intervals}.
Moreover, referring to the intervals $(q_{j-1}, q_j)$ with $j\in [r+1]$ as the {\em pre-fixed intervals},
we have that each round-$i$ interval contains $2^{l-i-1}$ pre-fixed intervals:
by construction,
each position added in round $i$ split the corresponding round-$(i-1)$ interval into two sub-intervals,
not with the same length but with the same number of pre-fixed intervals, thus
all round-$i$ intervals contain the same number of pre-fixed intervals.

Next, it is easy to see that the points that arrive in round $i$ are points $2^i+0, 2^i+1, \dots, 2^i + (2^i -1)$.
Thus point $d = 2^i + s$ arrives in round $i$ and the corresponding position $\tau_d$ is
in the round-$(i-1)$ interval $I^{i-1}_{s}$.
Accordingly,
there are $(2s+1)$
round-$i$ intervals to the left of $\tau_d$,
corresponding to
a total of $2^{l-i-1}(2s+1)$ pre-fixed intervals.
That is,
$\tau_d = q_{2^{l-i-1}(2s+1)}$ as we wanted to show.

Below we compute the minimum distance of the algorithm at time $d$.
Note that, after point $d$ is located, the intervals incurred by $\hat{Q}$ are
$I^i_0, I^i_1, \dots, I^i_{2s+1}, I^i_{2s+2}, I^{i-1}_{s+1}, \dots, I^{i-1}_{2^i-1}$.
By induction, we have that
\begin{itemize}
\item
the lengths of $I^{i}_0,I^{i}_1,\dots, I^{i}_{2s+2}$
are strictly decreasing,

\item
the lengths of $I^{i-1}_{s+1}, \dots, I^{i-1}_{2^i -1}$ are also strictly decreasing,

\item
$|I^{i-1}_{s}|= |I^i_{2s+1}|+|I^i_{2s+2}|$,
and

\item
$|I^{i-1}_{2^i -1}|< |I^{i-1}_{s}|< 2|I^{i-1}_{2^i -1}|$,
\end{itemize}
where the last inequality is because, for any two pre-fixed intervals
$(q_{j-1}, q_j)$ and $(q_{j'-1}, q_{j'})$ with $j< j'$,
we have
$|(q_{j'-1}, q_{j'})|< |(q_{j-1}, q_j)| < 2|(q_{j'-1}, q_{j'})|$.
Accordingly,
$|I^{i-1}_{2^i -1}|> \frac{|I^{i-1}_{s}|}{2} > |I^i_{2s+2}|$,
and the minimum distance at time $d$ is $d_{min}(d) = |I^i_{2s+2}|$.
As the left end-point of $I^i_{2s+2}$ is $\tau_d$ and the
 right end-point is $q_{2^{l-i-1}(2s+2)}$,
 we have
\begin{eqnarray*}
d_{min}(d) & = & q_{2^{l-i-1}(2s+2)} - q_{2^{l-i-1}(2s+1)} \\
& = & \left(\frac{1}{\sigma_r}\sum_{j=2^l}^{2^l-1+2^{l-i-1}(2s+2)} \frac{1}{j}\right)
- \left(\frac{1}{\sigma_r}\sum_{j=2^l}^{2^l-1+2^{l-i-1}(2s+1)} \frac{1}{j}\right)= \frac{1}{\sigma_r}\sum_{j = 2^l+2^{l-i-1}(2s+1)}^{2^l-1+2^{l-i-1}(2s+2)} \frac{1}{j},
\end{eqnarray*}
and Claim \ref{clm:1-d:formula} holds.
\end{proof}

%

\medskip

\begin{proof}[Proof of Claim \ref{clm:f(s)}]
For any $s<2^i-1$, we have

\begin{eqnarray*}
f(s+1) &=&
(2^{i}+s+2)\cdot
\sum\limits_{j=2^{l}+2^{l-i-1}(2s+3)}^{2^{l}-1+2^{l-i-1}(2s+4)}
\frac{1}{j}
= (2^{i}+s+2)\cdot
\sum\limits_{j=2^{l}+2^{l-i-1}(2s+1)+2^{l-i}}^{2^{l}-1+2^{l-i-1}(2s+2)+2^{l-i}} \frac{1}{j}\\
&=& (2^{i}+s+2) \cdot
\sum\limits_{j=2^{l}+2^{l-i-1}(2s+1)}^{2^{l}-1+2^{l-i-1}(2s+2)}\frac{1}{j+2^{l-i}}.
\end{eqnarray*}
It suffices to show that for all $s< 2^i-1$,
$f(s+1)-f(s)<0$:
that is,
\begin{equation}\label{equ:4}
f(s+1)-f(s)=
\sum\limits_{j=2^{l}+2^{l-i-1}(2s+1)}^{2^{l}-1+2^{l-i-1}(2s+2)}
\left(\frac{2^{i}+s+2}{j+2^{l-i}}-\frac{2^{i}+s+1}{j}\right)<0.
\end{equation}
Below we show that for each $j$ in the range of the summation,
$$\frac{2^{i}+s+2}{j+2^{l-i}}-\frac{2^{i}+s+1}{j}<0,$$
which is equivalent to
$(2^{i}+s+2)j<(j+2^{l-i})(2^{i}+s+1)$, or
$$j<2^{l-i}(2^{i}+s+1) = 2^l + 2^{l-i}(s+1) = 2^l + 2^{l-i-1}(2s+2).$$
However,
notice that the maximum value of $j$ in Equation \ref{equ:4} is
$2^{l}-1+2^{l-i-1}(2s+2)$,
thus all the inequalities above hold immediately.
Accordingly, Claim \ref{clm:f(s)} holds.
\end{proof}


%

\begin{proof}[Proof of Claim \ref{clm:g(i)}]
By definition,

\begin{eqnarray*}
g(i) &=& 2^{i+1}\cdot
\sum\limits_{j=2^{l+1}-2^{l-i-1}}^{2^{l+1}-1}\frac{1}{j}
= 2^{i+1}\cdot
\left(
\sum\limits_{j=2^{l+1}-2^{l-i-1}}^{2^{l+1}-2^{l-i-2}-1}\frac{1}{j} +
\sum\limits_{j=2^{l+1}-2^{l-i-2}}^{2^{l+1}-1}\frac{1}{j}
\right)\\
&=&
2^{i+1}\cdot
\left(
\sum\limits_{j=2^{l+1}-2^{l-i-2}}^{2^{l+1}-1}
\frac{1}{j-2^{l-i-2}}
+
\sum\limits_{j=2^{l+1}-2^{l-i-2}}^{2^{l+1}-1}\frac{1}{j}
\right) \\
&>& 2^{i+1} \cdot
\sum\limits_{j=2^{l+1}-2^{l-i-2}}^{2^{l+1}-1}\frac{2}{j}
= 2^{i+2} \cdot
\sum\limits_{j=2^{l+1}-2^{l-i-2}}^{2^{l+1}-1}\frac{1}{j}
= g(i+1),
\end{eqnarray*}
where the inequality
 is because
$\frac{1}{j-2^{l-i-2}}> \frac{1}{j}$ for all $j$ in the range of the summation.
Therefore Claim \ref{clm:g(i)} holds.
%
\end{proof}


\subsection{Proof for Theorem \ref{thm:1-dtight}}
\label{Appendix:Theorem2}


\paragraph{Theorem \ref{thm:1-dtight}.} (restated) {\em
There exists a deterministic polynomial-time online algorithm
for the ATWC problem, whose
competitive ratio can be arbitrarily close to
$2\ln 2$. Moreover, the running time is polynomial in $\frac{1}{\epsilon}$ for competitive ratio
$2\ln 2 + \epsilon$.
}

\begin{proof}
Note that $\sigma_{r}$ is strictly decreasing in $r$ and
$\lim_{l\to\infty}2\sigma_{2^l-1} = 2\ln 2$.
Thus for any small constant $\epsilon>0$, there exists $l$ such that
$2\sigma_{2^l-1}< 2\ln 2 + \epsilon$.
In particular,
it suffices to take
$l = \lceil \log_2 (\frac{2}{\epsilon}+1) -1 \rceil$.
The desired online algorithm first computes $l$ and $r = 2^l-1 \ (\leq \frac{2}{\epsilon})$,
then computes~$\tau$ and
runs Algorithm \ref{alg:1d-online} with $r$ and~$\tau$ on any online instance.

Since the selection of $l$ only depends on $\epsilon$ and not
on the input sequence, $l$ is a constant, and so is~$r$.
Given $l$ and $r$, the binary ordering $\tau$
can be constructed in time $O(r)= O(\frac{2}{\epsilon})$.
Given $r$ and~$\tau$,
when a point arrives or departs,
the running time of
Algorithm~\ref{alg:1d-online}
is polynomial in $|\hat{Q}|$,
thus polynomial in the
size of the input so far.
Accordingly, Theorem \ref{thm:1-dtight} holds.
%
\end{proof}

\subsection{Proof for Theorem \ref{thm:1-dexact}}
\label{Appendix:Theorem3}

\paragraph{Theorem \ref{thm:1-dexact}.} (restated) {\em
%
For any integer $d=2^i+s$ with $i\geq 0$ and $0\leq s\leq 2^i-1$, let
$\tau_d=
\frac{1}{\ln 2} \ln (1+\frac{2s+1}{2^{i+1}})=
\log_2 (1+\frac{2s+1}{2^{i+1}})$.
If Algorithm~\ref{alg:1d-online}
creates
the $d$-th new position in $\hat{Q}$ to be $\tau_d$,
 the
competitive ratio is exactly $2\ln 2$.
}

\begin{proof}
Similar to the proof of Lemma \ref{lem:1d-online},
the worst case happens
when each point arrives at different time and all points depart at the same time.
Arbitrarily fixing $d\geq 1$,
we consider
the minimum distance incurred by $\{\tau_1, \tau_2, \dots, \tau_d\}$ (at time $d$),
$d_{min}(d)$,
and the competitive ratio at time $d$,
$$apx(d) = \frac{1}{(d+1) d_{min}(d)}.$$
Similar to the proof of Claim \ref{clm:1-d:formula},
as $d$ grows,
$\frac{2s+1}{2^{i+1}}$ takes values $\frac{1}{2}, \frac{1}{4}, \frac{3}{4}, \frac{1}{8}, \frac{3}{8}, \frac{5}{8}, \frac{7}{8},\dots$ sequentially, following the breadth-first search on the infinite complete binary tree.
Therefore, when point $d$ is located,
the positions adjacent to
$\tau_d=\frac{1}{\ln 2} \ln (1+\frac{2s+1}{2^{i+1}})$ are
$d_{left} = \frac{1}{\ln 2} \ln (1+\frac{2s}{2^{i+1}})$ and
$d_{right} = \frac{1}{\ln 2} \ln (1+\frac{2s+2}{2^{i+1}})$,
with the minimum distance incurred by $\tau_d$ being either $d_{right}-\tau_d$ or $\tau_d - d_{left}$.
Note that, if $s = 0$
then $d_{left} = 0$,
which is the left end-point of the segment;
if $s=2^i-1$ then $d_{right}=1$,
which is the right end-point of the segment.
Moreover, if $1\leq s \leq 2^i-2$,
then
$\frac{2s}{2^{i+1}}$ and $\frac{2s+2}{2^{i+1}}$ can each
be written into
the form of
$\frac{2s'+1}{2^{j+1}}$ with $j\geq 0$ and $0\leq s'\leq 2^j-1$,
by taking out the greatest common diviser of the denominator and the enumerator.
Thus each of the two positions, $d_{left}$ and $d_{right}$,
corresponds to some position $\tau_{d'}$
where $d'=2^j+s'<d$.
Since
\begin{eqnarray*}
&& d_{right} - \tau_d = \frac{1}{\ln 2} \ln (1+\frac{2s+2}{2^{i+1}}) - \frac{1}{\ln 2} \ln (1+\frac{2s+1}{2^{i+1}})
= \frac{1}{\ln 2} \ln (1+\frac{1}{2^{i+1}+2s+1}) \\
&<& \frac{1}{\ln 2} \ln (1+\frac{1}{2^{i+1}+2s})
= \frac{1}{\ln 2} \ln (1+\frac{2s+1}{2^{i+1}}) - \frac{1}{\ln 2} \ln (1+\frac{2s}{2^{i+1}})
= \tau_d - d_{left},
\end{eqnarray*}
we have
$$d_{min}(d) = d_{right} - \tau_d = \frac{1}{\ln 2} \ln (1+\frac{1}{2^{i+1}+2s+1}).$$
Thus
\begin{eqnarray*}
apx(d) &=&
\frac{\ln2}{(2^{i}+s+1)\ln (1+\frac{1}{2^{i+1}+2s+1})} < \frac{\ln2}{(2^{i}+s+1)/(2^{i+1}+2s+2)} = 2\ln 2,
\end{eqnarray*}
where the inequality is because
$\ln(1+\frac{1}{x}) > \frac{1}{x+1}$ for any $x\geq 1$.
Moreover, when $d=2^i$ and $i\to \infty$,
$$
\lim_{i\to \infty} apx(2^i) = \frac{\ln2}{(2^{i}+1)\ln (1+\frac{1}{2^{i+1}+1})} = 2\ln 2.
$$
Therefore the desired competitive ratio is
$\sup_{d} apx(d) = 2\ln 2$
and
Theorem \ref{thm:1-dexact} holds.
\end{proof}

\paragraph{Remark:} Not only the algorithm  in Theorem \ref{thm:1-dexact} cannot be computed in polynomial time,
but it also cannot be properly approximated by just rounding each created position to a fixed precision.
The reason is that the online algorithm does not know $m$ beforehand and cannot adjust the precision according to it.
When the number of simultaneously present points grows,
 the optimal minimum distance may be much smaller than the precision,
the relative positions of the points may be completely
different, and the competitive ratio may be arbitrarily bad.
Instead, the polynomial-time algorithm in Algorithm \ref{alg:1d-online} deals with rational numbers and can adjust the precision as the number of present points grows.

\section{Proofs for Section \ref{sec:2D:ATWC}}\label{apx:2D}


\subsection{Proof for Lemma \ref{clm:boundsdisp}}
\label{appendix:clm:boundsdisp}

\paragraph{Lemma \ref{clm:boundsdisp}.} (restated) {\em
For any $n\geq 1$, $\frac{2}{5+\sqrt{2\sqrt{3}n}} \leq Disp(n; P) \leq \frac{2}{2+\sqrt{2\sqrt{3}n}}$,
where $P=[0, 1]^2$ is the unit square.
}

\begin{proof}
We first prove the upper-bound.
Recall that $DP(n; P)$ is
optimal radius for the
dispersal packing problem with $n$ balls.
By Claim \ref{clm:Disp-DP},
\begin{equation}
\label{eq:dp-disp}
 DP(n;P) = \frac{Disp(n;P)}{2(1-Disp(n;P))},
\end{equation}
as the radius of the insphere in the unit square $P$ is
$x = \frac{1}{2}$.
Because
the disk packing density in the unit square $P$
is upper-bounded by the disk packing density in $2$-dimensional infinite space,
which is
$\frac{\pi}{2\sqrt{3}}$~\cite{fejes1942dichteste},
we have $DP(n;P)\leq \frac{1}{\sqrt{2\sqrt{3}n}}$
by comparing the total area of
the $n$ disks packed in $P$ with $P$'s own area.
Accordingly,
\begin{equation}\label{eq:uperbpund:disp}
Disp(n; P) \leq \frac{2}{2+\sqrt{2\sqrt{3}n}}.
\end{equation}
In fact, $DP(n;P)\to \frac{1}{\sqrt{2\sqrt{3}n}}$ as $n\to\infty$,
thus Equation \ref{eq:uperbpund:disp} is tight as $n\to\infty$.

Next, we prove the lower-bound using
 hexagonal packing
in $P$, as illustrated in Figure \ref{fig:hex}.
In particular, we would like to find a radius $r^*$ such that,
each row of the packing
contains
$\left\lceil\sqrt{\frac{\sqrt{3}n}{2}}\right\rceil$
disks
 and
 there are $\left\lceil\sqrt{\frac{2n}{\sqrt{3}}}\right\rceil$ rows in total.
If so then we can pack $n$ disks in $P$
 with radius $r^*$, which implies
 $DP(n; P) \geq r^*$.

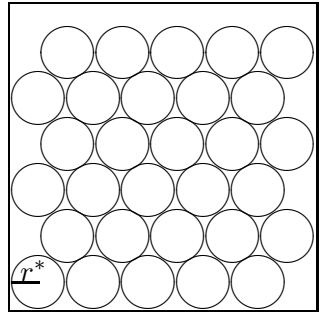
\begin{figure}[htbp]
\begin{center}
\setlength{\unitlength}{1cm}
\begin{picture}(4.1,4.1)

\thinlines

\multiput(0.39,0.39)(0.73,0){5}{\circle{0.75}}
\multiput(0.78,1)(0.73,0){5}{\circle{0.75}}
\multiput(0.39,1.61)(0.73,0){5}{\circle{0.75}}
\multiput(0.78,2.22)(0.73,0){5}{\circle{0.75}}
\multiput(0.39,2.83)(0.73,0){5}{\circle{0.75}}
\multiput(0.78,3.44)(0.73,0){5}{\circle{0.75}}

\put(0,0){\line(0,1){4.1}}
\put(0,0){\line(1,0){4.1}}
\put(4.1,0){\line(0,1){4.1}}
\put(0,4.1){\line(1,0){4.1}}

\put(0.03,0.38){\line(1,0){0.38}}
\put(0.15,0.4){$r^*$}

\end{picture}
\caption{Hexagonal packing.}
\label{fig:hex}
\end{center}
\end{figure}


%
%
%
%
%
%

Notice that, in order to pack $\left\lceil\sqrt{\frac{\sqrt{3}n}{2}}\right\rceil$
disks in a row,
the radius $r^*$
must satisfy
$$\left(2\left\lceil\sqrt{\frac{\sqrt{3}n}{2}}\right\rceil + 1\right)r^* \leq 1;$$
in order to pack
$\left\lceil\sqrt{\frac{2n}{\sqrt{3}}}\right\rceil$ rows,
the radius $r^*$ must satisfy
$$\left(\sqrt{3}\left(\left\lceil\sqrt{\frac{2n}{\sqrt{3}}}\right\rceil-1\right)+2\right) r^*\leq 1.$$
It is easy to verify that both inequalities are
satisfied by $r^*=\frac{1}{3+\sqrt{2\sqrt{3}n}}$.
Indeed,
$$\left(2\left\lceil\sqrt{\frac{\sqrt{3}n}{2}}\right\rceil + 1\right) \cdot \frac{1}{3+\sqrt{2\sqrt{3}n}}
\leq
\frac{2\sqrt{\frac{\sqrt{3}n}{2}} + 3}{3+\sqrt{2\sqrt{3}n}}=1$$
and
$$\left(\sqrt{3}\left(\left\lceil\sqrt{\frac{2n}{\sqrt{3}}}\right\rceil-1\right)+2\right)
\cdot \frac{1}{3+\sqrt{2\sqrt{3}n}}
\leq \frac{\sqrt{3}\cdot \sqrt{\frac{2n}{\sqrt{3}}}+2}{3+\sqrt{2\sqrt{3}n}}<1.$$
Accordingly, $DP(n; P) \geq \frac{1}{3+\sqrt{2\sqrt{3}n}}$
and
\begin{equation}\label{eq:lowerbound:disp}
Disp(n; P) \geq \frac{2}{5+\sqrt{2\sqrt{3}n}}
\end{equation}
by Equation \ref{eq:dp-disp}.
Thus Lemma \ref{clm:boundsdisp} holds.
\end{proof}

\subsection{Proof for Theorem \ref{thm:2d-lower}}
\label{appendix:thm:2d:lower}

\paragraph{Theorem \ref{thm:2d-lower}.} (restated) {\em
No online algorithm achieves a competitive ratio better than 1.183 for the 2-dimensional ATWC problem
in a square.
}

\begin{proof}
%
Arbitrarily fixing an online algorithm $\cal A$ and denoting its competitive ratio by $\sigma$,
we show that $\sigma \geq 1.183$. To do so,
arbitrarily fix a positive integer $r$
and consider the set of positive integers
$$\Delta = \left\{\delta' \ :\  \frac{2}{5+\sqrt{2\sqrt{3}(r+\delta')}}
\geq \frac{1}{2+\sqrt{2\sqrt{3}r}}\right\}.$$
Solving the inequality for $\delta'$, we have
%
$$
\delta' \leq \left\lfloor\frac{(2\sqrt{2\sqrt{3}r}-1)^{2}}{2\sqrt{3}} -r\right\rfloor =
\left\lfloor 3r-4\sqrt{\frac{r}{2\sqrt{3}}}+ \frac{1}{2\sqrt{3}}\right\rfloor.
$$
Letting $\delta = \max \Delta$, we have
$$\delta = \left\lfloor 3r-4\sqrt{\frac{r}{2\sqrt{3}}}+ \frac{1}{2\sqrt{3}}\right\rfloor
\ \mbox{ and }\ \Delta = \{1, \dots, \delta\}.$$
Following Lemma \ref{clm:boundsdisp}, it is easy to see that
$Disp(r+\delta;P) \geq \frac{Disp(r;P)}{2}$.
We now consider the input sequence
$S = ((1, r+\delta + 1),(2, r+\delta + 1),\dots,(r+\delta, r+\delta+1))$ and
 prove
\begin{equation}\label{eq:2d:lb}
r\pi \left(\frac{Disp(r;P)}{2\sigma} \right)^2
+\sum\limits_{i=1}^{\delta}
\pi\left(
\frac{Disp(r+i;P)}{\sigma} -\frac{Disp(r;P)}{2\sigma}
\right)^2
\leq \left(1-\frac{Disp(r;P)}{\sigma}\right)^2.
\end{equation}
Indeed, after the first $r$ points arrive,
$d_{min}(r) \geq \frac{Disp(r; P)}{\sigma}$ by assumption,
thus we can pack~$r$ disks with radius $\frac{Disp(r; P)}{2\sigma}$
in the square, centered at the $r$ positions where the points are located.
Their total area is
$$S_1 = r\pi \left(\frac{Disp(r;P)}{2\sigma}\right)^2.$$

Next, for each arriving point $r+i$ with $1\leq i\leq \delta$,
the distance between itself and any
other point $j$ with $1\leq j \leq r$ is at least

\begin{equation}\label{equ:8+1}
d_{min}(r+i) \geq \frac{Disp(r+i;P)}{\sigma}
= \left(\frac{Disp(r+i;P)}{\sigma} -\frac{Disp(r;P)}{2\sigma}\right) + \frac{Disp(r;P)}{2\sigma},
\end{equation}
and
the distance between itself and any
other point $r+j$ with $1\leq j <i$ is at least

\begin{eqnarray}\label{equ:8+2}
&& d_{min}(r+i) \geq \frac{Disp(r+i;P)}{\sigma} \nonumber \\
&=&
\left(\frac{Disp(r+i;P)}{\sigma} -\frac{Disp(r;P)}{2\sigma}\right) + \frac{Disp(r;P)}{2\sigma} \nonumber \\
&\geq&  \left(\frac{Disp(r+i;P)}{\sigma} -\frac{Disp(r;P)}{2\sigma}\right) +
\left(\frac{Disp(r+j;P)}{\sigma} -\frac{Disp(r;P)}{2\sigma}\right),
\end{eqnarray}
where the last inequality is because $Disp(n; P)$ is non-increasing in $n$ and
$Disp(r; P) \geq Disp(r+j; P)$.
By the definition of $\delta$, for any $i\leq \delta$ we have
$$\frac{Disp(r+i;P)}{\sigma} -\frac{Disp(r;P)}{2\sigma}
\geq \frac{Disp(r+\delta;P)}{\sigma} -\frac{Disp(r;P)}{2\sigma}
\geq 0,$$
 and the left-hand side of the first inequality is a well-defined radius.
Accordingly,
by Equations \ref{equ:8+1} and \ref{equ:8+2},
if we put a disk
with radius
$\frac{Disp(r+i;P)}{\sigma} -\frac{Disp(r;P)}{2\sigma}$
centered at the position of point $r+i$ for each $1\leq i \leq \delta$,
then
the disk $r+i$ does not overlap with the first $r$ disks
whose radius is $\frac{Disp(r; P)}{2\sigma}$,
neither does it overlap with
any disk $r+j$ with $1\leq j< i$, whose radius is $\frac{Disp(r+j;P)}{\sigma} -\frac{Disp(r;P)}{2\sigma}$.
Moreover,
since the radius of disk $r+i$ is at most $d_{min}(r+i)$,
it does not overlap with the boundary of $P$ either.
By induction,
all  $r+\delta$ disks do not overlap with each other or with the boundary of~$P$,
and they are a {\em non-uniform packing} in $P$.
The total area of all the disks $r+i$ with $1\leq i \leq \delta$ is
$$S_2 = \sum\limits_{i=1}^{\delta}
\pi\left(
\frac{Disp(r+i;P)}{\sigma} -\frac{Disp(r;P)}{2\sigma}
\right)^2.$$


Finally,
since
$$d_{min}(r) - \frac{Disp(r; P)}{2\sigma} \geq
\frac{Disp(r; P)}{2\sigma}
$$
and
$$d_{min}(r+i) - \left(\frac{Disp(r+i;P)}{\sigma} -\frac{Disp(r;P)}{2\sigma}\right)
\geq
\frac{Disp(r;P)}{2\sigma} \ \ \forall  1\leq i\leq \delta,$$
these $r+\delta$ disks are actually packed within
the square
$[\frac{Disp(r; P)}{2\sigma}, 1-\frac{Disp(r; P)}{2\sigma}]^2$.
Thus
$$S_1+ S_2 \leq \left(1-\frac{Disp(r;P)}{\sigma}\right)^2$$
and Inequality \ref{eq:2d:lb} holds.
%

Replacing
$Disp(r+i;P)$ and $Disp(r;P)$ in
the left-hand side of Inequality \ref{eq:2d:lb}
by proper lower- and upper-bounds
from Lemma \ref{clm:boundsdisp}
and
moving $\sigma$ to the right-hand side of the inequality, we have
\begin{eqnarray*}
& & r\pi \left(\dfrac{1}{5+\sqrt{2\sqrt{3}r}}\right)^2
+\sum\limits_{i=1}^{\delta}\pi\left(
\dfrac{2}{5+\sqrt{2\sqrt{3}(r+i)}} - \dfrac{1}{2+\sqrt{2\sqrt{3}r}}
\right)^2 \\
& = & r\pi \left(\dfrac{1}{5+\sqrt{2\sqrt{3}r}}\right)^2
+ \sum\limits_{i=1}^{\delta}\pi\left(\dfrac{2}{5+\sqrt{2\sqrt{3}(r+i)}}\right)^2
+ \sum\limits_{i=1}^{\delta}\pi\left(\dfrac{1}{2+\sqrt{2\sqrt{3}r}}\right)^2 \\
& & -\sum\limits_{i=1}^{\delta}
\dfrac{4 \pi}{\left(5+\sqrt{2\sqrt{3}(r+i)}\right)\left(2+\sqrt{2\sqrt{3}r}\right)} \\
&\leq& \sigma^2 \left(1-\frac{Disp(r;P)}{\sigma}\right)^2.
\end{eqnarray*}
Replacing $\delta$ with
$\left\lfloor 3r-4\sqrt{\frac{r}{2\sqrt{3}}}+ \frac{1}{2\sqrt{3}}\right\rfloor$
and letting $r\to\infty$,
we calculate the limits term by term as follows.

\begin{eqnarray*}
&&\lim\limits_{r\to \infty} r\pi \left(\dfrac{1}{5+\sqrt{2\sqrt{3}r}}\right)^2 = \frac{\pi}{2\sqrt{3}}; \\
&&\lim\limits_{r\to \infty} \sum\limits_{i=1}^{\delta}
\pi\left(\dfrac{2}{5+\sqrt{2\sqrt{3}(r+i)}}\right)^2
\geq \lim\limits_{r\to \infty}
\int_{r+1}^{r+\delta+1}
\pi\left(\dfrac{2}{5+\sqrt{2\sqrt{3}x}}\right)^2 dx \\
& & \quad \quad = \lim\limits_{r\to\infty} \frac{4\pi}{\sqrt{3}}
\ln \left(\frac{\sqrt{2\sqrt{3}(r+\delta+1)}+5}{\sqrt{2\sqrt{3}(r+1)}+5}\right)
= \lim\limits_{r\to\infty} \frac{4\pi}{\sqrt{3}} \ln \left(\frac{\sqrt{2\sqrt{3}\cdot 4r}}{\sqrt{2\sqrt{3}r}}\right)
= \frac{4\pi \ln 2}{\sqrt{3}}; \\
&&\lim\limits_{r\to \infty}
\sum\limits_{i=1}^{\delta}
\pi \left(\dfrac{1}{2+\sqrt{2\sqrt{3}r}}\right)^2
= \lim\limits_{r\to \infty}
\pi \delta \left(\dfrac{1}{2+\sqrt{2\sqrt{3}r}}\right)^2
= \lim\limits_{r\to \infty}
 \dfrac{3\pi r}{2\sqrt{3}r}
= \frac{\sqrt{3}\pi}{2}; \\
&&\lim\limits_{r\to \infty}
\sum\limits_{i=1}^{\delta}
\dfrac{4\pi}{\left(5+\sqrt{2\sqrt{3}(r+i)}\right) \left(2+\sqrt{2\sqrt{3}r}\right)}
\leq \lim\limits_{r\to \infty}
\sum\limits_{i=1}^{\delta}
\dfrac{4\pi}{\sqrt{2\sqrt{3}(r+i)} \cdot \sqrt{2\sqrt{3}r}} \\
& &  \quad \quad  = \frac{2\pi}{\sqrt{3}}\lim\limits_{r\to \infty}
\sum\limits_{i=1}^{\delta}
\dfrac{1}{\sqrt{r(r+i)}}
\leq \frac{2\pi}{\sqrt{3}} \lim\limits_{r\to \infty} \int_{0}^{\delta} \frac{1}{\sqrt{r(r+x)}}dx
= \frac{2\pi}{\sqrt{3}} \lim\limits_{r\to \infty} \frac{2}{\sqrt{r}}(\sqrt{r+\delta} - \sqrt{r}) \\
& &  \quad \quad  = \frac{2\pi}{\sqrt{3}}  \lim\limits_{r\to \infty} 2\left(\sqrt{1+\frac{\delta}{r}}-1\right)
= \frac{4\pi}{\sqrt{3}}; \\
&& \lim\limits_{r\to \infty}
\sigma^2 \left(1-\frac{Disp(r;P)}{\sigma}\right)^2 = \sigma^2.
\end{eqnarray*}
Accordingly,
$$\frac{\pi}{2\sqrt{3}} + \frac{4\pi\ln 2}{\sqrt{3}} + \frac{\sqrt{3}\pi}{2} - \frac{4\pi}{\sqrt{3}}
= \frac{2(2\ln 2 - 1)\pi}{\sqrt{3}}\leq \sigma^2$$
 and
 $\sigma \geq \sqrt{\frac{2(2\ln 2 - 1)\pi}{\sqrt{3}}} \approx 1.183$.
Therefore Theorem \ref{thm:2d-lower} holds.
\end{proof}

\subsection{Algorithm \ref{alg:2dcreation}}
\label{appendix:alg:2dcreation}
In this section we define Algorithm \ref{alg:2dcreation}, the Position Creation Phase used by Algorithm \ref{alg:2donline}.
It proceeds in rounds $i = 0, 1, \dots $, and Figure \ref{fig:round-0} illustrates
the colored areas after round 0.

\begin{algorithm}[htbp]\nonumber
\caption{{\hspace{-3pt}{\bf .}} The Position Creation Phase.}
  \label{alg:2dcreation}

In this phase, our algorithm
repeatedly splits the
rectangles into
smaller ones
by creating a position at its center.
The internal areas (green, pink and red)
will expand and the outer areas (orange, yellow and blue)
will still be the strips with width of a single rectangle.

\hspace{10pt}
After $q_{36}$ is added to $\hat{Q}$,
the position creation phase proceeds in rounds $i = 0, 1, \dots$ such that, after
each round, each of the previous rectangles has been divided into
four sub-rectangles by a position at its center.
Accordingly, the number of rectangles created in round $i\geq 0$,
referred to as {\em round-$i$} rectangles,
is $49\cdot 4^{i+1}$; the number of intervals in each dimension is $7\cdot 2^{i+1}$;
and the number of positions in $\hat{Q}$ is $(7\cdot 2^{i+1}-1)^2$, corresponding to
the grid vertices not on the boundary.
Moreover, the positions created in round $i$
are positions $(7\cdot 2^{i}-1)^2+1\leq n \leq (7\cdot 2^{i+1}-1)^2$.

\hspace{10pt}
The algorithm keeps the round number $i$ and updates as it proceeds. (Or, given $n>36$,
we can find in time $O(\log n)$ the round $i$ in which position $n$ is created
from the formula above.)
The location of position $n$ is decided by the following five cases.
We provide the ranges of position $n$ for each case,
so that it can be easily decided which case position $n$ belongs to.

\begin{enumerate}[{Case} 1.]
\item
If there is a round-$(i-1)$ square in the green area
whose center is not in $\hat{Q}$, arbitrarily
pick such a rectangle and put position $n$ at its center,
splitting it into four round-$i$
rectangles.

We claim these are positions $(7\times 2^i-1)^2+1\leq n\leq = 65\times 2^{2i}-22\times2^i+2$.

\item
Else, if there is a round-$(i-1)$ rectangle in the pink areas
whose center is not in $\hat{Q}$,
arbitrarily pick such a rectangle and put position $n$ at its center,
splitting it into four round-$i$
rectangles.

We claim these are positions
$65\times 2^{2i}-22\times2^i+3 \leq n \leq 89\times 2^{2i}-36\times2^i+4$.

%
%
%
%
\item
Else, if there exists a round-$(i-1)$ rectangle in the red area
 (which is actually a square)
whose center
is not in $\hat{Q}$, arbitrarily
pick such a rectangle and create position $n$ at its center,
splitting it into four round-$i$
rectangles (again, squares).

We claim these are positions
$89\times 2^{2i}-36\times2^i+5 \leq n \leq 98\times2^{2i}-42\times2^i+5$.

\item
Else, we distinguish two sub-cases.

First, if there exists a round-$(i-1)$ rectangle in the orange area and
the two blue areas whose center
is not in
$\hat{Q}$, arbitrarily
pick such a rectangle and create position $n$ at its center,
splitting it into four round-$i$
rectangles.

Otherwise, if there exists a vertex of a round-$i$ rectangle in the green and the orange
areas which
is (a) not in $\hat{Q}$,
(b) not
adjacent to the blue or the pink areas and (c) not on the boundary,
then arbitrarily
pick such a vertex for position $n$.

We claim these are positions
 $98\times2^{2i}-42\times2^i+6\leq n \leq 130\times2^{2i}-36\times2^i+2$.

\item
Else, among all the centers of the round-$(i-1)$ rectangles in the yellow area
and all the vertices of the round-$i$ rectangles, which
is not in $\hat{Q}$ and not on the boundary,
arbitrarily pick one for position $n$.

We claim these are positions
$130\times2^{2i}-36\times2^i+3 \leq n \leq (7\times 2^{i+1}-1)^2$.
\end{enumerate}

After all positions in round $i$ are created,
the yellow, blue, and orange areas shrink by one round-$i$ rectangle
towards the boundary: they only contain rectangles adjacent to the boundary.
The area released by orange is taken by green;
that released by blue is taken by pink;
and that released by yellow is taken by blue, pink, and red.
Figure \ref{fig:round-0} illustrates
the colored areas after round~0.
\end{algorithm}

\begin{figure}[htbp]
\begin{center}
\begin{tikzpicture}

\fill [color=green](0,0) rectangle (3.6,3.6);
\fill [color=red](0,0) rectangle (-2,-2);
\fill [color=pink](0,0) rectangle (-2,3.6);
\fill [color=pink](0,0) rectangle (3.6,-2);
\fill [color=yellow](-3,-3) rectangle (-2,4.8);
\fill [color=yellow](-3,-3) rectangle (4.8,-2);
\fill [color=orange](3.6,0) rectangle (4.8,4.8);
\fill [color=orange](0,3.6) rectangle (4.8,4.8);
\fill [color=blue](-2,3.6) rectangle (0,4.8);
\fill [color=blue](3.6,0) rectangle (4.8,-2);

\fill [color=green](0,0) rectangle (4.2,4.2);
\fill [color=red](-2,0) rectangle (-2.5,-2.5);
\fill [color=red](0,-2) rectangle (-2.5,-2.5);
\fill [color=pink](0,0) rectangle (-2.5,4.2);
\fill [color=pink](0,0) rectangle (4.2,-2.5);
\fill [color=blue](-2.5,4.2) rectangle (0,4.8);
\fill [color=blue](4.2,0) rectangle (4.8,-2.5);

\draw[xstep=.5cm,ystep=.5] (0,0) grid (-3,-3);
\draw[xstep=.5cm,ystep=.6] (0,0) grid (-3,4.8);
\draw[xstep=.6cm,ystep=.5] (0,0) grid (4.8,-3);
\draw[xstep=.6cm,ystep=.6] (0,0) grid (4.8,4.8);

\path (0.2,0.2) node {$1$};
\path (2.6,2.6) node {$2$};
\path (0.2,2.6) node {$3$};
\path (2.6,0.2) node {$4$};
\path (1.4,1.4) node {$5$};
\path (0.2,0.2) node {$1$};
\path (1.4,-0.8) node {$6$};
\path (-0.8,1.4) node {$7$};
\path (-0.8,-0.8) node {$8$};
\path (2.6,1.4) node {$9$};
\path (1.4,2.6) node {$10$};
\path (3.8,3.8) node {$11$};
\path (3.8,2.6) node {$12$};
\path (3.8,1.4) node {$13$};
\path (3.8,0.2) node {$14$};
\path (2.6,3.8) node {$15$};
\path (1.4,3.8) node {$16$};
\path (0.2,3.8) node {$17$};
\path (0.2,1.4) node {$18$};
\path (1.4,0.2) node {$19$};
\path (2.6,-0.8) node {$20$};
\path (-0.8,2.6) node {$21$};
\path (-0.8,0.2) node {$22$};
\path (0.2,-0.8) node {$23$};
\path (3.8,-0.8) node {$24$};
\path (3.8,-1.8) node {$25$};
\path (2.6,-1.8) node {$26$};
\path (1.4,-1.8) node {$27$};
\path (0.2,-1.8) node {$28$};
\path (-0.8,-1.8) node {$29$};
\path (-1.8,-1.8) node {$30$};
\path (-1.8,-0.8) node {$31$};
\path (-1.8,0.2) node {$32$};
\path (-1.8,1.4) node {$33$};
\path (-1.8,2.6) node {$34$};
\path (-1.8,3.8) node {$35$};
\path (-0.8,3.8) node {$36$};

%

\end{tikzpicture}

\caption{The round-0 rectangles and colored areas for Algorithm \ref{alg:2dcreation}.
The positions $q_1, \dots, q_{36}$ are still labelled,
in order to compare with Figure \ref{fig:36p:2nd}.}
\label{fig:round-0}
\end{center}
\end{figure}
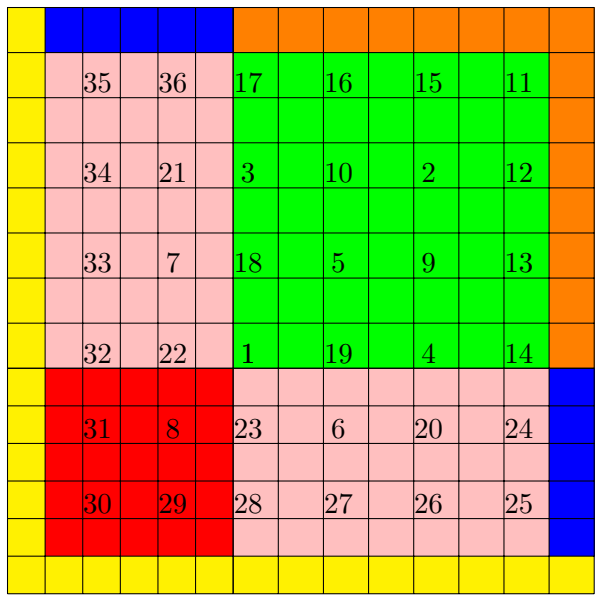

\subsection{Some Intuition for the Set $Q$ of Pre-fixed Positions}
\label{sec:intuition}

To help understanding the performance of
Algorithm \ref{alg:2donline},
we first provide
some intuition for
the selection of the first 36 positions in $Q$.
Similar to those for Algorithm \ref{alg:1d-online},
they cannot be too unevenly distributed,
and cannot be too evenly distributed either.
%
That is why
$q_1$ is not
at the center of the square.
Without loss of generality, it is closer to the left and the bottom boundary.
Letting $y$ be its distance
to the left (and the bottom) boundary,
$\frac{Disp(1; P)}{y} = \frac{1}{2y}$
will be roughly the competitive ratio.
Again similar to Algorithm \ref{alg:1d-online},
$q_1$'s distance to the right and the top boundary, $1-y$, is larger than $y$
but smaller than $2y$.

To make the minimum distance shrink as little as possible,
$q_2$ is put in the center of the larger square, to the upper-right of $q_1$.
Note that
the best position for $q_2$ given $q_1$
is of
equal distance to
$q_1$ and the boundary,
 as illustrated by Figure \ref{fig:point2}. We instead put it in the center of the upper-right square
 so as to leave enough space for $q_3$ and $q_4$.
 Indeed, $q_3$ and $q_4$ are symmetric and at the grid vertices
 induced by $q_1$ and $q_2$, so that adding them to $\hat{Q}$ does not
 shrink the minimum distance.

\begin{figure}[htbp]
\begin{minipage}[t]{0.5\textwidth}
\begin{center}
\setlength{\unitlength}{0.5cm}
\thicklines
\begin{picture}(8.6,8.6)
\put(0,0){\line(0,1){8}}

\put(3,0){\line(0,1){8}}
\put(5,5){\line(0,1){3}}

\put(8,0){\line(0,1){8}}

\put(0,0){\line(1,0){8}}

\put(0,3){\line(1,0){8}}
\put(5,5){\line(1,0){3}}

\put(0,8){\line(1,0){8}}

\put(3,3){\line(1,1){2}}

\put(2.3,2.3){$1$}
\put(1.2, 2.5){$y$}
\put(5.2,5.2){$2$}
\put(6,4.3){$d$}
\put(4.4,6){$d$}
\put(4.3,3.7){$d$}

\end{picture}
\caption{The best position for $q_2$ given $q_1$.}
\label{fig:point2}
\end{center}
\end{minipage}
\begin{minipage}[t]{0.5\textwidth}
\begin{center}
\setlength{\unitlength}{0.5cm}
\thicklines
\begin{picture}(8.6,8.6)
\put(0,0){\line(0,1){8}}

\put(3,0){\line(0,1){8}}
\put(0,5.5){\line(1,0){8}}
\put(5.5,0){\line(0,1){8}}

\put(8,0){\line(0,1){8}}

\put(0,0){\line(1,0){8}}

\put(0,3){\line(1,0){8}}

\put(0,8){\line(1,0){8}}


\put(2.3,2.3){$1$}
\put(5.6,5.6){$2$}
\put(2.4,5.6){$3$}
\put(5.6,2.4){$4$}

\put(1.3,6.3){$A$}
\put(4.0,6.3){$B$}
\put(6.6,6.3){$C$}
\put(1.3,3.8){$D$}
\put(4.0,3.8){$E$}
\put(6.6,3.8){$F$}
\put(1.3,1.3){$G$}
\put(4.0,1.3){$H$}
\put(6.6,1.3){$I$}

\end{picture}
\caption{After the first four positions.}
\label{fig:point5}
\end{center}
\end{minipage}
\end{figure}

After the first four positions,
 the square is split
into 9 small rectangles, from $A$ to $I$, as shown in Figure \ref{fig:point5}.
Let us now consider $q_5$.
Notice that among the four upper-right squares,
$E$ is better than $B, C$ and $F$,
as the latter three have the same size with $E$
but are adjacent to the boundaries.
If we put $q_5$ at the center of $E$,
then eventually we will slice the area to the upper-right of $q_1$
into four intervals in each dimension, all of length $\frac{1-y}{4}$.

Another possibility is to put $q_5$ in $G$.
However, note that the center of $G$ is not the
optimal position in $G$, for the same reason
as the choice of $q_2$.
We will choose $y$ such that
the center of $E$ is better than the optimal position in $G$.
In fact, a better competitive ratio is achieved
if we slice $G$ into
three intervals instead of four in each dimension, all of length $\frac{y}{3}$.
Let $x = \frac{y}{3}$ and $c x = \frac{1-y}{4}$ be the
corresponding lengths of the intervals after all the slicing,
as illustrated by Figure \ref{fig:36p:2nd}.
We have $x = \frac{1}{3+4c}$
and the minimum distance incurred by putting $q_5$ at the center of $E$ is
$\sqrt{2}cx$.

After the grid with 36 internal vertices is created,
the positions on the grid
appear naturally by always selecting the next optimum.
In order to find this optimal position
and compute the corresponding minimum distance easily,
in a 1-dimensional space we would choose $x< cx < 2x$ as in Algorithm \ref{alg:1d-online}.
In a 2-dimensional space we instead want $x<cx< \sqrt{2}x$, that is, $1<c< \sqrt{2}$,
because the distances are $L_2$-norm.
The resulting ordering for $Q$ is shown in Figure \ref{fig:36p:2nd}
and the
minimum distances after each position $n$ are shown in Table~\ref{table:Q}, grouped
into 7 cases.

\begin{table}[htbp]
\renewcommand{\arraystretch}{1.2}
\begin{center}
\begin{tabular}{|c|c|c|c|c|c|c|c|}
\hline
Case&1&2&3&4&5&6&7\\ \hline
$n$ & $1$ & $2,3,4$ & $5$ & $6,7$ &$8$ & $9,\cdots,17$ & $18,\cdots, 36$ \\ \hline
$d_{min}(n)$ & $3x$ & $2cx$ & $\sqrt{2}cx$ & $\sqrt{c^{2}+1}x$ & $\sqrt{2}x$ & $cx$ & $x$ \\  \hline
\end{tabular}
\end{center}

\caption{The minimum distances as positions in $Q$ are created.}
\label{table:Q}
\end{table}

\subsection{Proof for Theorem \ref{thm:2d-worst}}
\label{App:thm:2d-worst}

\paragraph{Theorem \ref{thm:2d-worst}.} (restated) {\em
Algorithm \ref{alg:2donline} runs in polynomial time and
it is $1.591$-competitive for the 2-dimensional online ATWC problem in a square.
}
\medskip

We first show that the algorithms are well defined; see the claim below.

\begin{claim}\label{clm:position-n}
The ranges of position $n$ for each case in Algorithm \ref{alg:2dcreation} are all correct.
\end{claim}

\begin{proof}
For each $i\geq 0$, recall that
the positions created in round $i$ are
$(7\cdot 2^{i}-1)^2+1\leq n \leq (7\cdot 2^{i+1}-1)^2$.
Below we show how the positions are distributed across
the five cases of Algorithm~\ref{alg:2dcreation},
where for $i=0$, round $i-1$ refers to the creation of the first 36 positions in $Q$.

\begin{enumerate}[{Case} 1.]
\item
Denote the number of round-$(i-1)$
squares in the green area by $N_g^{i-1}$.
By construction,
the number of round-$i$ squares in the green area
is $N_g^i = (2\sqrt{N_g^{i-1}}+1)^2$:
indeed, $\sqrt{N_g^{i-1}}$ is the number of intervals
in each dimension in the green area after round $i-1$;
$2\sqrt{N_g^{i-1}}$ is obtained because each one of them is split into two after round $i$;
and the extra $+1$ is because, at the end of round $i$,
 the orange area shrinks and the green area grows by one round-$i$ rectangle
 to the top and to the right.

Solving this inductive formula with initial condition $N_g^{-1} = 9$,
we have
$$
N_g^{i-1} = (4\times 2^i - 1)^2
$$
 for each $i\geq 0$,
which is exactly the number of positions
created in Case 1 of round $i$,
because one position is created
at the center of every such square.
Accordingly, in Case 1,
$$(7\cdot 2^{i}-1)^2+1\leq n \leq (7\cdot 2^{i}-1)^2 + (4\times 2^i - 1)^2
= 65\times 2^{2i}-22\times2^i+2.$$

\item
Let $N_p^{i-1}$ be the number of round-$(i-1)$ rectangles in the pink areas.
In order to compute $N_p^{i-1}$, we first
compute $N_r^{i-1}$,
the number of round-$(i-1)$ squares in the red area.
Similar to the induction above,
we have $N_r^{-1} = 4$ and
$N_r^i = (2\sqrt{N_r^{i-1}}+1)^2$,
thus
$$
N_r^{i-1} = (3\times 2^i-1)^2.
$$
Accordingly,
$$N_p^{i-1} = 2\times \sqrt{N_r^{i-1}} \times \sqrt{N_g^{i-1}} =
2(3\times 2^i-1)(4\times 2^i-1),$$
which is exactly the number of positions
created in Case 2 of round $i$,
because one position is created
at the center of every such square.
Thus in Case 2,
$$(7\cdot 2^{i}-1)^2 + (4\times 2^i - 1)^2 +1 \leq n \leq
(7\cdot 2^{i}-1)^2 + (4\times 2^i - 1)^2 + 2(3\times 2^i-1)(4\times 2^i-1),$$
that is,
$$65\times 2^{2i}-22\times2^i+ 3 \leq n \leq
89\times 2^{2i}-36\times2^i+4.$$

\item
Similarly, in Case 3 of
round $i$,
the number of positions created is
$N_r^{i-1} = (3\times 2^i-1)^2$ and
$$89\times 2^{2i}-36\times2^i+5 \leq n \leq 89\times 2^{2i}-36\times2^i+4 + (3\times 2^i-1)^2
= 98\times2^{2i}-42\times2^i+5.$$

\item
Furthermore,
the number of round-$(i-1)$ orange rectangles and that of round-$(i-1)$ blue rectangles are, respectively,
$$N_o^{i-1} =
(\sqrt{N_g^{i-1}}+1)^2 - N_g^{i-1} =
(4\times 2^i)^2 - (4\times 2^i - 1)^2 = 8\times 2^{i}-1$$
and
$$N_b^{i-1} = 2\times \sqrt{N_r^{i-1}} = 2(3\times 2^i-1).$$
Thus the first sub-case in Case 4 of round $i$
creates $N_o^{i-1} + N_b^{i-1} =
(8\times 2^{i}-1)+(6\times 2^{i}-2) = 14\times 2^i -3$ positions.

In the second sub-case, there are
$N_g^{i-1} + N_o^{i-1}$ round-$(i-1)$ squares in the green and the orange areas combined,
thus the same number of centers created in them so far.
For each of the upper-left $N_g^{i-1}$ round-$(i-1)$
squares, its center induces 2 new positions,
one to its right and one to its bottom.
For each of the remaining $N_o^{i-1}-1$ round-$(i-1)$ squares
with the one at the bottom-right corner excluded,
its center induces 1 new position,
either to the bottom or to the right, because the other one is either on the boundary
or is adjacent to the pink area below.
Moreover, the center of the
orange round-$(i-1)$ square at the bottom-right corner does not induce any new position, because
the one to its right is on the boundary and the one to its bottom is adjacent to the blue area below.
Accordingly, the number of new positions
 created in this subcase is
 $2N_g^{i-1} + N_o^{i-1}-1 =
 2(4\times 2^i - 1)^2 + 8\times 2^{i}-1 - 1 =
 32\times 2^{2i}-8\times 2^{i}$.

Combining the two sub-cases together,
the positions created in Case 4 of round $i$ are
$$98\times2^{2i}-42\times2^i+6\leq n \leq
98\times2^{2i}-42\times2^i+5 + (14\times 2^i -3) +  (32\times 2^{2i}-8\times 2^{i})$$
\hspace{108pt}  $= 130\times2^{2i}-36\times2^i+2.$

\item
Finally, all the remaining positions
in round $i$ are created in Case 5, with
$$130\times2^{2i}-36\times2^i+3 \leq n \leq (7\cdot 2^{i+1}-1)^2.$$
For completeness, we note that the number of
round-$(i-1)$ yellow rectangles is
$$N_y^{i-1} = (7\times 2^i)^2 - (7\times 2^i-1)^2 = 14\times 2^{i}-1.$$
\end{enumerate}
Combining all the cases together, Claim \ref{clm:position-n} holds.
\end{proof}

Finally we  are ready to prove Theorem~\ref{thm:2d-worst}.


\begin{proof}[Proof of Theorem \ref{thm:2d-worst}]
It is easy to see that
the running times of Algorithms~\ref{alg:2donline} and \ref{alg:2dcreation} are polynomial in $|\hat{Q}|$ and thus
in the number of arrived points.
For the competitive ratio,
again we can focus on instances of
the form
$S=((1,n+1),(2,n+1),\dots,(n,n+1))$,
with $n\geq 1$.
%
Letting $apx(n;P)=\frac{Disp(n; P)}{d_{min}(n)}$,
we would like to find
$c = \arg\min\limits_c \max\limits_n apx(n;P)$.
In the analysis below
we distinguish three cases:
first $n\leq 36$, then $n$ in rounds 0 and~1 of Algorithm \ref{alg:2dcreation},
and at last $n$ in round $i$ for each $i\geq 2$ of Algorithm \ref{alg:2dcreation}.

%

\paragraph{For $n\leq 36$.}
Following Table \ref{table:Q} in Section \ref{sec:intuition},
the minimum distances within each of the seven cases are the same.
Since $Disp(n ;P)$ is non-increasing,
the worst competitive ratios occur
at the first position~$n$ in each case.
For example,
the worst competitive ratio in case~2 is $\frac{Disp(2; P)}{2cx}$.

For $n = 1, 2$ and $5$, the exact solution for $Disp(n; P)$ can be easily found.
That is,
with one position in the center for $n=1$;
with two positions on a diagonal for $n=2$, where the ratios of the three intervals
are $\sqrt{2}:1:\sqrt{2}$;
and with one position in the center and four positions in the four resulting squares
for $n=5$, where the four positions are all on the diagonals and
the ratios of the resulting intervals on a diagonal is
$\sqrt{2}:1:1:\sqrt{2}$.

Although there is no general closed-form for $Disp(n; P)$ in the literature,
exact solutions for $DP(n; P)$ have been found for small $n$'s, and we use them
to get the worst competitive ratios of our algorithm for $6\leq n\leq 36$.
Extending Table \ref{table:Q}, the competitive ratios are shown in Table~\ref{table:ratioQ}.
Note that here we do not need the upper bound of $Disp(n;P)$, as we know the exact solution.



\begin{table}[htbp]
\renewcommand{\arraystretch}{1.3}
\begin{center}
\begin{tabular}{|c|c|c|c|c|c|c|c|}
\hline
Case&1&2&3&4&5&6&7\\ \hline
$n$ & $1$ & $2,3,4$ & $5$ & $6,7$ &$8$ & $9,\cdots,17$ & $18,\cdots, 36$ \\ \hline
$d_{min}(n)$ & $3x$ & $2cx$ & $\sqrt{2}cx$ & $\sqrt{c^{2}+1}x$ & $\sqrt{2}x$ & $cx$ & $x$ \\  \hline
$apx(n;P)$ & $\frac{0.5}{3x}$ & $\frac{0.36940}{2cx}$ & $\frac{0.29290}{\sqrt{2}cx}$ & $\frac{0.27292}{\sqrt{c^{2}+1}x}$ \cite{schwartz1970separating} & $\frac{0.25434}{\sqrt{2}x}$ \cite{schaer1965geometric} & $\frac{1}{4cx}$ \cite{schaer1965densest} & $\frac{0.18769}{x}$ \cite{peikert1992packing} \\
\hline
\end{tabular}
\end{center}
\caption{Competitive ratios for $n\leq 36$. 
Note that
\cite{schwartz1970separating, schaer1965geometric, schaer1965densest, peikert1992packing}
provide exact solutions for the corresponding $DP(n; P)$'s.
We compute the corresponding $Disp(n;P)$'s by Claim \ref{clm:Disp-DP}.}
\label{table:ratioQ}
\end{table}



%

\paragraph{Round $0$.}
Our construction of Algorithm \ref{alg:2dcreation} guarantees that,
in each round $i\geq 0$,
$d_{min}(n)$ is the same for all $n$'s in the same case.
The general formulas of $d_{min}(n)$ in all five cases are shown in Table~\ref{table:round0},
and the
competitive ratios
are obtained by
applying our upper-bound for
$Disp(n; P)$ in Lemma \ref{clm:boundsdisp}.
Again note that the worst competitive ratio in each case happens at the first position~$n$,
because $Disp(n; P)$ is non-decreasing.
The corresponding ratios for round 0 can be obtained by setting $i=0$.


\begin{table}[htbp]
\renewcommand{\arraystretch}{1.5}
\begin{center}
\begin{tabular}{|c|c|c|c|}
\hline
Case & $n$ & $d_{min}(n)$ & $apx(n;P)$\\ \hline
1 &
\begin{tabular}{l}
  $n\geq (7\times 2^i-1)^2+1$ and   \\
  $n \leq 65\times 2^{2i}-22\times2^i+2 $
\end{tabular}
 & {$\dfrac{\sqrt{2}cx}{2^{i+1}}$}
 & {$\frac{2\sqrt{2}\times2^{i}}{cx(2+\sqrt{2\sqrt{3}((7\times 2^i-1)^2+1)})}$}\\ \hline
2 &
\begin{tabular}{l}
$n \geq 65\times 2^{2i}-22\times2^i+3$ and  \\
$n \leq 89\times 2^{2i}-36\times2^i+4$
\end{tabular}
 & {$\dfrac{x\sqrt{1+c^2}}{2^{i+1}}$}
 & {$\frac{4\times 2^{i}}{x\sqrt{1+c^{2}}(2+\sqrt{2\sqrt{3}(65\times 2^{2i}-22\times2^i+3)})}$} \\ \hline
3 &
\begin{tabular}{l}
$n \geq 89\times 2^{2i}-36\times2^i+5 $ and \\
$n \leq 98\times2^{2i}-42\times2^i+5$
\end{tabular}
 & $\dfrac{\sqrt{2}x}{ 2^{i+1}}$
 & $\frac{2\sqrt{2}\times 2^{i}}{x(2+\sqrt{2\sqrt{3}(89\times 2^{2i}-36\times2^i+5)})}$\\ \hline
4 &
\begin{tabular}{l}
$n \geq 98\times2^{2i}-42\times2^i+6$ and  \\
$n \leq 130\times2^{2i}-36\times2^i+2$
\end{tabular}
 & $\dfrac{cx}{2^{i+1}} $
 & $\frac{4\times 2^{i}}{cx(2+\sqrt{2\sqrt{3}(98\times2^{2i}-42\times2^i+6)})}$ \\ \hline
5 &
\begin{tabular}{l}
$n \geq 130\times2^{2i}-36\times2^i+3 $ and \\
$n \leq (7\times 2^{i+1}-1)^2$
\end{tabular}
 & $\dfrac{x}{2^{i+1}}$
 & $\frac{4\times 2^{i}}{x(2+\sqrt{2\sqrt{3}(130\times2^{2i}-36\times2^i+3)})}$\\ \hline

\end{tabular}
\end{center}

\caption{Competitive ratios for round $i\geq 0$.}
\label{table:round0}
\end{table}
Now we show the correctness of $d_{min}(n)$ in
Table \ref{table:round0} for any $i\geq 0$.
If $n$ falls into Case 1, it is
at
the center of a round-$(i-1)$ square
in the green area, and
the length of the edge is
$\frac{4cx}{\sqrt{N_g^{i-1}+1}} = \frac{4cx}{4\times 2^i} = \frac{cx}{2^{i}}$.
Since the minimum distance is half of the diagonal of that square,
we have
$$d_{min}(n)=\frac{cx}{\sqrt{2} 2^{i}}.$$

If $n$ falls into Case 2,
it is at the center of
a round-$(i-1)$ rectangle in
the pink areas.
The length of the rectangle is
$\frac{cx}{2^{i}}$ and the width is
$\frac{3x}{\sqrt{N_r^{i-1}}+1} = \frac{3x}{3\times 2^i} =
\frac{x}{2^{i}}$.
Since the minimum distance is
half of the diagonal of that rectangle,
we have
$$d_{min}(n)=\frac{\sqrt{1+c^{2}}}{2}\times\frac{x}{2^{i}}=\frac{x\sqrt{1+c^2}}{2^{i+1}}.$$

If $n$ falls into Case 3,
it is at the center of a
round-$(i-1)$ square in the red area,
 and the length of the edge is
 $\frac{x}{2^{i}}$.
Again, the minimum distance is half of the diagonal of that square and we have
$$d_{min}(n)=\frac{x}{\sqrt{2} 2^{i}}.$$

If $n$ falls into Case 4,
then it
is either
at the center of a round-$(i-1)$ rectangle in the orange and the blue areas,
 or
it is a vertex of a round-$i$ rectangle in the green and the orange areas
such that it is not on the boundary or adjacent to the pink or the blue areas.
No matter which sub-case happens,
the minimum distance incurred by position $n$ is $\frac{1}{2}\times\frac{cx}{2^{i}}$,
either to the boundary
or to the center of a round-$(i-1)$ rectangle next to it.
Thus
$$d_{min}(n)=\frac{cx}{2^{i+1}}.$$

Finally, if $n$ falls into Case 5,
it is either at the center of a round-$(i-1)$ rectangle in yellow area or
it is
a vertex of a round-$i$ rectangle in the blue,
yellow, pink or red areas.
No matter which sub-case happens,
the minimum distance incurred by position $n$
is $\frac{1}{2}\times\frac{x}{2^{i}}$, either
to the boundary
or to the center of the round-$(i-1)$ rectangle next to it.
Therefore
$$d_{min}(n)=\frac{x}{2^{i+1}}.$$
In sum, the general formulas of $d_{min}(n)$ for any $i$ and $n$
are as shown in Table \ref{table:round0}.

\paragraph{Round 1.}
By setting $i=1$ in Table \ref{table:round0},
similar formulas for $apx(n; P)$ can be obtained for all five cases of round 1.

\paragraph{Round $i \geq 2$.} 
To find the worst competitive ratio for all rounds $i\geq 2$,
the difficulty is that the upper-bounds for $apx(n; P)$ in
Table \ref{table:round0} are not necessarily monotone, thus it is hard to tell
where the worst ratio occurs.
Instead,
we find a universal upper-bound for the ratios
and a novel way to analyze its monotonicity.
More precisely,
let $f(n;P)=\frac{2}{\sqrt{2\sqrt{3}n}\cdot d_{min}(n)}$.
By Lemma \ref{clm:boundsdisp}, for any $n\geq 1$,
$$apx(n;P)=\dfrac{Disp(n;P)}{d_{min}(n)}
\leq \frac{2}{(2+\sqrt{2\sqrt{3}n})d_{min}(n)}
\leq \frac{2}{\sqrt{2\sqrt{3}n}\cdot d_{min}(n)}
=f(n;P).$$
In Table \ref{table:roundi},
we replace $axp(n;P)$ by
its upper-bound
$f(n;P)$.

\begin{table}[htbp]
\renewcommand{\arraystretch}{1.5}
\begin{center}
\begin{tabular}{|c|c|c|c|}
\hline
Case & $n$ & $d_{min}(n)$ & $f(n;P)$\\ \hline
1 &
\begin{tabular}{l}
  $n\geq (7\times 2^i-1)^2+1$ and  \\
  $n \leq 65\times 2^{2i}-22\times2^i+2 $
\end{tabular}
 & $\dfrac{\sqrt{2}cx}{2^{i+1}}$
 & $\frac{2\sqrt{2}\times2^{i}}{cx\sqrt{2\sqrt{3}((7\times 2^i-1)^2+1)}}$\\ \hline
2 &
\begin{tabular}{l}
$n \geq 65\times 2^{2i}-22\times2^i+3$ and \\
$n \leq 89\times 2^{2i}-36\times2^i+4$
\end{tabular}
 &$\dfrac{x\sqrt{1+c^2}}{2^{i+1}}$
 &$\frac{4\times 2^{i}}{x\sqrt{1+c^{2}}\sqrt{2\sqrt{3}(65\times 2^{2i}-22\times2^i+3)}}$ \\ \hline
3 &
\begin{tabular}{l}
$n \geq 89\times 2^{2i}-36\times2^i+5 $ and \\
$n \leq 98\times2^{2i}-42\times2^i+5$
\end{tabular}
 & $\dfrac{\sqrt{2}x}{2^{i+1}}$
 & $\frac{2\sqrt{2}\times 2^{i}}{x\sqrt{2\sqrt{3}(89\times 2^{2i}-36\times2^i+5)}}$\\ \hline
4 &
\begin{tabular}{l}
$n \geq 98\times2^{2i}-42\times2^i+6$ and \\
$n \leq 130\times2^{2i}-36\times2^i+2$
\end{tabular}
 & $\dfrac{cx}{2^{i+1}} $
 & $\frac{4\times 2^{i}}{cx\sqrt{2\sqrt{3}(98\times2^{2i}-42\times2^i+6)}}$ \\ \hline
5 &
\begin{tabular}{l}
$n \geq 130\times2^{2i}-36\times2^i+3 $ and\\
$n \leq (7\times 2^{i+1}-1)^2$
\end{tabular}
& $\dfrac{x}{2^{i+1}}$ &
   $\frac{4\times 2^{i}}{x\sqrt{2\sqrt{3}(130\times2^{2i}-36\times2^i+3)}}$\\ \hline
\end{tabular}
\end{center}

\caption{Competitive ratios upper-bounded by $f(n; P)$.}
\label{table:roundi}
\end{table}
%

Given Tables \ref{table:round0} and \ref{table:roundi},
for any $n\geq 37$,
it falls into one entry
 based on its round number~$i$ and case number $j\in \{1, \dots, 5\}$.
Accordingly,
we also use $apx(i,j)$ and $f(i,j)$ to denote $apx(n;P)$  and $f(n;P)$.
When $1\leq n\leq 36$,
it falls into one of the seven cases in Table \ref{table:ratioQ},
and we denote the corresponding quantities by $apx(*,j)$ and $f(*,j)$ with $j \in \{1, \dots, 7\}$.
Since we want to find a proper parameter $c$ to minimize
$\max\{\max\limits_{i\geq 0,\ j\leq 5}apx(i,j),
\max\limits_{j\leq 7}apx(*,j)\} $,
we have to find where $apx(i,j)$ or $apx(*,j)$ is maximized.

When $i\geq 2$,
given any $j\leq 5$,
it is easy to see that
$f(i,j)$ is decreasing with respect to $i$---that is,
$f(n; P)$ may not be monotone overall, but it is monotone
across the same case!
Therefore we only need to compare
$apx(*,j)$ for $j=1,\cdots,7$, $apx(i,j)$ for $i = 0, 1$ and $j=1,\cdots,5$,
 and $f(2,j)$ for $j=1,\cdots,5$.
Note that $f(i, j)$ is also decreasing at $i=0, 1$,
but it is a loose upper-bound there, thus
we use $apx(i, j)$ for $i=0, 1$ in order to get a better bound.



What remains is simple.
Indeed, there are in total $22$ $apx(\cdot; \cdot)$'s and $f(\cdot; \cdot)$'s
to consider,
whose values can be directly calculated
from the tables.
More specifically,
each one of them is less
than or equal to $\max\{apx(*,6), apx(0,5)\}$.
Notice that, so far, we haven't used the value of $c$.
By choosing $c$ to minimize $\max\{apx(*,6), apx(0,5)\}$,
we get
$c = \frac{2+\sqrt{194\sqrt{3}}}{16} \approx 1.271 \in (1, \sqrt{2})$,
\footnote{In order to deal with irrational numbers,
we round $c$ to $1.271$ and the result still holds.}
and the two values are both $\frac{1}{4cx}$.
Accordingly, the competitive ratio
of Algorithm \ref{alg:2donline} is
$\max\limits_n\frac{Disp(n;P)}{d_{min}(n)} = \frac{1}{4cx} < 1.591$,
and Theorem \ref{thm:2d-worst} holds.
%
\end{proof}

\section{Proofs for Section \ref{sec:kdatwc}}
\label{appendix:offline:ATWC}
%



\paragraph{Theorem \ref{thm:kd-lower}.} (restated) {\em
For any $k \geq 2$, no algorithm achieves a competitive ratio better
than $\frac{7}{6}$ for the online ATWC problem for arbitrary polytopes.
}

\begin{proof}
We show that even when the given $k$-dimensional ``polytope'' $P$ is a sphere and there are only two points, no algorithm achieves a competitive ratio better than $\frac{7}{6}$.
As a sphere can be approximated arbitrarily closely by a polytope, this implies our theorem.

Without loss of generality, let the center of the sphere
be $(0,0,\dots,0)$ and the radius be~$\frac{1}{2}$.
We first consider
the optimal solution for
the dispersion problem without time.
Clearly, $Disp(1;P) = \frac{1}{2}$.
When there are two points, we have the following.

\begin{claim}\label{clm:kd-lower-off}
For any algorithm $\cA$ locating 2 points in $P$,
there exists another
%
algorithm $\cA'$ such that $\cA'$ locates
the two points on a diameter of the sphere
and the competitive ratio of $\cA'$ to $Disp(2; P)$ is no worse than $\cA$.
\end{claim}

\begin{proof}
Without loss of generality,  assume $\cA$
locates point 1 at $p_1 = (x,0,0,\dots,0)$ with $x\leq 0$.
Denote the position of point 2 according to $\cA$ as $p_2=(y_1,y_2,\dots,y_k)$.
It is easy to see that $p_{2}$ has distance  $y'=\sqrt{\sum_i^k y_i^2}$ to the center and $\frac{1}{2}-y'$ to the boundary.
Let algorithm $\cA'$ locate point~1 at $p_{1}$ and point 2 at $p'_2=(y',0,0,\dots,0)$.
Note that $dis(p_2', \partial P) = \frac{1}{2}-y' = dis(p_2, \partial P)$
and $dis(p'_2, p_1) = y'-x \geq dis(p_2, p_1)$, by triangle inequality.
Thus Claim \ref{clm:kd-lower-off} holds.
\end{proof}

%

Following Claim \ref{clm:kd-lower-off},
it is not hard to see that
the optimal solution has $p_1$ and $p_2'$ with
$x=-\frac{1}{6}$ and $y'=\frac{1}{6}$, thus
$Disp(2;P)=\frac{1}{3}$.
Now we consider online algorithms for ATWC.
Following Claim~\ref{clm:kd-lower-off},
without loss of generality we focus on algorithms
that locate the second point on the same diameter as the first.
Again without loss of generality,
the two positions are
$p_1=(x,0,\dots,0)$ and $p_2=(y,0,\dots,0)$, with $x\leq 0\leq y$.
Moreover, given $p_1$,
the best position for $p_2$ is to set $y$ such that
 $dis(p_2, \partial P) = dis(p_2, p_1)$.
That is, $y - x = \frac{1}{2} - y$, which implies $y = \frac{1+2x}{4}$.
Accordingly, $d_{min}(2; P)=\min \left\{\frac{1}{2}+x, \frac{1-2x}{4}\right\}$
and the worst competitive ratio after the first two points is
$$\max \left\{\dfrac{\frac{1}{2}}{\frac{1}{2}+x},\dfrac{\frac{1}{3}}{\min \left\{\frac{1}{2}+x, \frac{1-2x}{4}\right\}}\right\}.$$
Choosing $x$ to minimize this maximum,
we have $x=-\frac{1}{14}$ and the competitive ratio cannot be better than $\frac{7}{6}$.
Thus Theorem \ref{thm:kd-lower} holds.
\end{proof}

Similarly, a lower-bound of $1+\frac{1}{4\sqrt{k}+2}$ can be obtained for $k$-dimensional cubes for any~$k\geq~2$.
To construct an online algorithm for ATWC for arbitrary polytopes,
we first consider the straightforward greedy algorithm, and we have the following.
Recall that the geometric problems
 of finding the minimum bounding cube, deciding whether a position is in~$P$,
 and finding the distance between a point in $P$ and the boundary of $P$
 are given as oracles.

\begin{lemma}\label{lem:exist}
For any $k\geq 2$, $n\geq 1$, $k$-dimensional polytope $P$ and $n$ positions $p_1,\dots, p_n\in P$,
there exists a position $p\in P$ such that $\min_{i\in [n]} \{dis(p, p_i), dis(p, \partial P)\}\geq \frac{Disp(n;P)}{2}$.
\end{lemma}

\begin{proof}
Denote by $\cP$ the set of points in $P$ with distance larger than or equal to $\frac{Disp(n;P)}{2}$ from the boundary.
By contradiction, assume that for every $p \in \cP$, $\min_{i\in [n]} \{dis(p, p_i)\} < \frac{Disp(n;P)}{2}$.
Letting $r =\sup_{p\in \cP} \min_{i\in [n]} \{dis(p, p_i)\}$, we have $r< \frac{Disp(n;P)}{2}$, because $\cP$ is closed.
Accordingly, $\cP$ is covered by $n$ balls centered at $p_1,\dots, p_n$ with radius $r$
and
$$vol(\cP) \leq n Ball(r),$$
 where $vol(\cP)$ is the volume of $\cP$ and $Ball(r)$ is the volume of a ball with radius $r$.
However, by the definition of $Disp(n; P)$, there exist $n$ positions in $P$ such that
$$\min_{i, j\in [n]}\{dis(p_i, p_j), dis(p_i, \partial P)\} \geq Disp(n; P).$$
Thus $n$ balls with radius $\frac{Disp(n;P)}{2}$ can be packed into $\cP$, which implies
$$vol(\cP)\geq n Ball(\frac{Disp(n;P)}{2})$$
and we get a contradiction, because $Ball(r) < Ball(\frac{Disp(n;P)}{2})$.
Therefore Lemma \ref{lem:exist} holds.
\end{proof}

Consider the greedy algorithm that, for each $n\geq 1$,
creates the $(n+1)$-st position at $p = \arg\max_{p\in P} \min_{i\in [n]} \{dis(p, p_i), dis(p, \partial P)\}$.
(As Algorithm \ref{alg:1d-online}, it adds $p$ to $\hat{Q}$ and creates a position only if positions in $\hat{Q}$ are all occupied.)
Since $\frac{Disp(n;P)}{2} \geq \frac{Disp(n+1;P)}{2}$,
by Lemma~\ref{lem:exist} and an inductive reasoning,
it is easy to see that this online algorithm is 2-competitive for the ATWC problem.
However, finding the optimal position $p$
 may be time consuming even given the oracles, and
we design a polynomial-time competitive algorithm which achieves a competitive ratio of $\frac{2}{1-\epsilon}$ for any $\epsilon>0$.

\paragraph{Theorem \ref{thm:kd:approv}.} (restated) {\em
%
For any constants $\gamma, \epsilon>0$, for any integer $k\geq 2$ and any $k$-dimensional polytope $P$ with covering rate at least $\gamma$,
there exists a deterministic polynomial-time online algorithm for the ATWC problem,
with competitive ratio $\frac{2}{1-\epsilon}$ and running time polynomial in~$\frac{1}{(\gamma\epsilon)^k}$.
}

\begin{proof}
Without loss of generality, the minimum bounding cube of $P$ is the unit cube.
Thus the edge-length of the maximum inscribed cube $C$ is at least $\gamma$.
The idea is to simply
slice the unit cube into small cubes and
exhaustively search all the cube-centers that are in $P$.
The number of cubes need to be searched depends on the number $n$ of existing positions,
the approximation parameter~$\epsilon$,
and a lower-bound for $Disp(n; P)$, which ultimately depends on~$\gamma$.
%
In particular, we have the following.

\begin{claim}\label{lem:slice}
For any $n\geq 1$, let $m = \left\lceil\frac{\sqrt{k}(n^{1/k}+2)}{\gamma\epsilon}\right\rceil$.
If we uniformly
slice the unit cube into $m^k$ smaller cubes $\{C_1,C_2,\dots, C_{m^k}\}$ with edge-length $\frac{1}{m}$,
then for any $n$ positions $p_1, \dots, p_n \in P$,
there exists a cube $C_j$ whose center $c_j$ is in $P$ and $\min_{i\in [n]}\{dis(c_j, p_i), dis(c_j, \partial P)\}\geq \frac{(1-\epsilon)Disp(n;P)}{2}$.
\end{claim}

\begin{proof}
Let $p$ be the optimal position chosen by the greedy algorithm given $p_1, \dots, p_n$,
$C_j$ be the cube containing $p$,  and $c_j$ the center of $C_j$.
We have that
\begin{equation}\label{equ:7-1}
 \min_{i\in [n]}\{dis(p, p_i), dis(p, \partial P)\} \geq \frac{Disp(n;P)}{2}
\end{equation}
 and
\begin{equation}\label{equ:7-2}
dis(p, c_j)\leq \frac{\sqrt{k}}{2m}\leq \frac{\epsilon \gamma}{2(n^{1/k}+2)}.
\end{equation}

Now we show $Disp(n;P) \geq \frac{\gamma}{n^{1/k}+2}$
by finding $n$ positions $p'_1, \dots, p'_n\in P$ such that
$$\min_{i, j\in [n]}\{dis(p'_i, p'_j), dis(p'_i, \partial P)\}\geq \frac{\gamma}{n^{1/k}+2}.$$
In particular, we uniformly slice cube $C$ into
$(\lceil n^{1/k}\rceil+1)^k$ smaller cubes with
edge-length at least $\frac{\gamma}{\lceil n^{1/k}\rceil+1}$
and arbitrarily choose $n$ vertices of them from the interior of $C$: note that there are $(\lceil n^{1/k}\rceil)^k\geq n$ such vertices.
It is easy to see that
$$\min_{i, j\in [n]}\{dis(p'_i, p'_j), dis(p'_i, \partial P)\}
\geq \min_{i, j\in [n]}\{dis(p'_i, p'_j), dis(p'_i, \partial C)\}
\geq \frac{\gamma}{\lceil n^{1/k}\rceil+1} \geq \frac{\gamma}{n^{1/k}+2},$$
which implies
\begin{equation}\label{equ:7-3}
Disp(n;P) \geq \frac{\gamma}{n^{1/k}+2}
\end{equation}
 as desired.
Combining Equations \ref{equ:7-1}, \ref{equ:7-2}, \ref{equ:7-3}, for any $\epsilon< 1$, we have $dis(p, c_j) < \frac{\gamma}{2(n^{1/k}+2)} \leq  \frac{Disp(n;P)}{2} \leq dis(p, \partial P)$,
thus $c_j\in P$.
Moreover, for any $i\in [n]$, $dis(p, c_j) < \frac{Disp(n;P)}{2}\leq dis(p, p_i)$.
By the triangle inequality,
$$dis(c_j, p_i) \geq dis(p, p_i) - dis(p, c_j) \geq \frac{Disp(n;P)}{2} -  \frac{\gamma\epsilon}{2(n^{1/k}+2)} \geq \frac{(1-\epsilon)Disp(n;P)}{2}.$$
Similarly,
$$dis(c_j, \partial P) \geq dis(p, \partial P) - dis(p, c_j)\geq  \frac{(1-\epsilon)Disp(n;P)}{2}.$$
Accordingly,
$$\min_{i\in [n]}\{dis(c_j, p_i), dis(c_j, \partial P)\} \geq \frac{(1-\epsilon)Disp(n;P)}{2}$$
as desired, and Claim \ref{lem:slice} holds.
\end{proof}


Thus by going through all the $c_j$'s in $P$ and choosing the one that maximizes \\
$\min_{i\in [n]}\{dis(c_j, p_i), dis(c_j, \partial P)\}$, we find the $(n+1)$-st position in time $O(m^k\cdot n) = O(\frac{n^2 k^{k/2}}{\gamma^k\epsilon^k})$.
Again because
$Disp(n; P) \geq Disp(n+1;P)$, applying Claim \ref{lem:slice} to each round of the greedy algorithm,
by induction we have that the resulting algorithm runs in polynomial time and is $\frac{2}{1-\epsilon}$-competitive.
Therefore Theorem \ref{thm:kd:approv} holds.
\end{proof}

\begin{corollary}
\label{col:onlineCD}
For any $k$-dimensional polytope $P$ with covering rate at least
$\gamma$,
the algorithm in Theorem \ref{thm:kd:approv} is $O(n)$-competitive
for the online CD problem.
\end{corollary}

\begin{proof}
Arbitrarily fix a small constant $\epsilon$ in this greedy algorithm.
For any $S = ((s_1, d_1),\dots, (s_n, d_n))$ and
at any time $0\leq t \leq T$,
the output of the algorithm satisfies $d_{min}(t)\geq \frac{(1-\epsilon)OPT_A(S; P)}{2} \geq \frac{(1-\epsilon)Disp(n;P)}{2}\geq \frac{(1-\epsilon)\gamma}{2n}$.
Thus $\int_{0}^{T}d_{min}(t) \geq \frac{(1-\epsilon)\gamma T}{2n}$.
Since $OPT_{C}(S;P) = \max_{X_1, \dots, X_n} \int_0^T d_{min}(t; X) dt$$ \leq T\cdot 1 =T$, the competitive ratio is $O(n)$.
\end{proof}

\section{Proofs for Section \ref{sec:kdcd}}
\label{appendix:offline:cd}

\begin{claim}\label{claim:A_I}
The output of Algorithm $\cal A _{I}$ satisfies properties $\Phi.1$  and $\Phi.2$ and the running time of $\cal A _{I}$ is $O(n^{2})$.
\end{claim}

\begin{proof}
Given an input sequence $S = ((s_{1},d_{2}),...,(s_{n},d_{n}))$, let $\cT$ be the union of all the intervals that contain at least one point.
Note that, for the original input sequence, we may assume $\cT = [0, T]$ without loss of generality. However, this may change in later iterations.

In each round of $\cal A _{I}$'s {\bf while} loop,
the algorithm considers all the points that arrive within the sliding window $(s, d]$ and depart after the sliding window.
If there are such points, then it selects the point $j$ with the latest departure time
 among them and moves the sliding window
to $(d, d_j]$. It puts point $j$ in $\cI_{index}$, which alternates between $\cI_1$ and $\cI_2$ in different rounds.
If there is no such point, then it finds the earliest arriving time after $d$, $s' = \min\limits_{i\in S, s_i>d} s_i$, and moves the sliding window to $(d, s']$.
Note that whichever case happens, the new sliding window satisfies $s<d$, thus is well defined.

Now we prove the output $\{\cI_1, \cI_2\}$  satisfies $\Phi.1$.
Indeed, note that the variable $index$ has the same value for any two rounds $l$ and $l+2$. Let $(s, d]$ be the sliding window in round $l$.
If point $j$ is added to $\cI_{index}$ in round $l$ and point $j'$ is added to $\cI_{index}$ in round $l+2$,
then the three sliding windows in rounds $l, l+1$ and $l+2$ are, respectively,
$(s, d], (d, d_j]$ and $(d_j, x]$ with some $x> d_j$, computed either in Step \ref{step:8-8} or Step \ref{step:8-11}.
By the definition of $j'$, $s_{j'}> d_j$ and points $j$ and $j'$ do not overlap.
By induction, we have that no two points in $\cI_1$ overlap, neither do any two points in $\cI_2$.
Thus property $\Phi.1$ holds.

Next, we prove the output $\{\cI_1, \cI_2\}$ satisfies $\Phi.2$: in other words, the time intervals of points in $\cI = \cI_1\cup \cI_2$ cover $\cT$.
To do so, note that the sliding windows in Algorithm ${\cal A_I}$ are all disjoint
and their union covers $[0, T]$. Accordingly, for any time $t\in [0, T]$ with at least one point $j'$ at time $t$ in $S$,
there exists a unique round $i$ with $t$ in its sliding window $(s^i, d^i]$: that is,
\begin{equation}\label{equ:8-1}
s_{j'}\leq t\leq d^i \mbox{ and } d_{j'}\geq t > s^i.
\end{equation}
We distinguish three cases.
\begin{enumerate}[{Case} 1.]
\item
If $s^i=-1$ (i.e., round $i$ is the first round in the algorithm), then it must be $d^i = 0 = s_{j'} = t$ and $d_{j'}>0$.
By definition, $\hat{S}\neq \emptyset$ and the chosen point $j$ is such that $s_j = 0$ and $d_j\geq d_{j'}$. Therefore at least one point at time $t$ (i.e., point $j$) is in $\cI$.

\item
$s^i\geq 0$ and the sliding window $(s^i, d^i]$ is obtained from round $i-1$ in Step \ref{step:8-8}.
Letting $j$ be the point chosen in round $i-1$ in Step \ref{step:8-6}, we have $s_j \leq s^i$ and $d_j = d^i$, thus $j$ is a point present at time $t$ in $S$ and $j\in \cI$.

\item
$s^i\geq 0$ and the sliding window $(s^i, d^i]$ is obtained from round $i-1$ in Step \ref{step:8-11}.
Letting $(s^{i-1}, d^{i-1}]$ be the sliding window in round $i-1$. By the construction of the algorithm,
we have $\hat{S} = \emptyset$ in round $i-1$ and $s^i = d^{i-1}$, thus
$d_{j'}> d^{i-1}$ by the second part of Equation~\ref{equ:8-1}. Below we consider the three possible intervals for $s_{j'}$.

If $s_{j'}> d^{i-1}$, then $d^i\leq s_{j'}$ by Step \ref{step:8-11}, thus $d^i = s_{j'} = t$ by the first part of Equation \ref{equ:8-1}.
Since $d_{j'}> s_{j'}$, we have $j'\in \hat{S}$ in round $i$, thus some point $j$ is chosen in Step \ref{step:8-6} of round $i$ with $s_j\leq d^i$ and $d_j> d^i$.
Therefore $j$ is present at time $t$ and $j\in \cI$, as desired.

If $s^{i-1}< s_{j'}\leq d^{i-1}$, then $j'\in \hat{S}$ in round $i-1$, which is a contradiction. Thus $j'$ cannot be in this interval.

Finally, if $s_{j'}\leq s^{i-1}$, then there exists another round $i'$ and sliding window $(s^{i'}, d^{i'}]$ with $i'<i-1$, which contains $s_{j'}$.
Since $d^{i'}\leq s^{i-1} < d^{i-1}$, we have $d_{j'}> d^{i'}$ and $j'\in \hat{S}$ in round~$i'$.
Accordingly, there is a point $j$ chosen in round $i'$ and $d^{i'+1} = d_j$.
However,
$d_{j'} \geq t> s^i = d^{i-1} \geq d^{i'+1} = d_j$, contradicting the fact that $j$ has the latest departure time among all points in $\hat{S}$ in round $i'$.
Thus $j'$ cannot be in this interval either.

Accordingly, in Case 3 we have $s_{j'}> d^{i-1}$ and there is a point in $\cI$ which is present at $t$.
\end{enumerate}
Combining all the cases, property $\Phi.2$ holds.

%
%

Finally, because there are $2n$ arriving and departure times in the sequence $S$ and the end-points of the sliding windows only occur at those times,
algorithm $\cal A_I$ uses at most $2n$ sliding windows (including $(-1, 0]$) and at most $2n$ rounds.
Since each round  takes $O(n)$ time, the running time of $\cal A_I$ is $O(n^2)$.
%
Therefore Claim \ref{claim:A_I} holds.
\end{proof}

%
%
%
%
%
%
%
%


\paragraph{Theorem \ref{thm:integral}.} (restated) {\em
For any $k\geq 1$ and $k$-dimensional polytope $P$, given any polynomial-time online algorithm $\mathcal{A}_{ATWC}$ for the ATWC problem with competitive ratio $\sigma$, there is
a polynomial-time offline algorithm $\cA_{CD}$ for the CD problem with competitive ratio $\sigma\max\limits_{i\geq 1}\frac{Disp(i;P)}{Disp(2i;P)} \leq 2\sigma$, using $\mathcal{A}_{ATWC}$ as a black-box.
}

\begin{proof}

We first consider the running time of $\cA_{CD}$.
Each time it runs $\cal A_{I}$, the number of points in $S$ decreases by at least one,
thus $\cA_{CD}$ runs for at most $n$ rounds in the {\bf while} loop: that is, $r\leq n$ in the end.
By Claim~\ref{claim:A_I}, it takes $O(n^3)$ time to finish the {\bf while} loop of $\cA_{CD}$.
Therefore the total running time of $\cA_{CD}$ is $O(n^3 + Time(\cA_{ATWC}))$,
where the second term is the running time of the blackbox algorithm ${\cal A}_{ATWC}$ on $2r\leq 2n$ points.
Since ${\cal A}_{ATWC}$ runs in polynomial time, $\cA_{CD}$ runs in polynomial time.

Below we analyze the competitive ratio of $\cA_{CD}$.
After the {\bf while} loop of $\cA_{CD}$,
$S$ is partitioned into $\{\cI_{1},\cI_{2}\},\dots, \{\cI_{2r-1},\cI_{2r}\}$.
Note that for each $i\in \{0, 1, \dots, r-1\}$,
by Claim~\ref{claim:A_I}, $\{\cI_{2i+1},\cI_{2i+2}\}$ satisfies property $\Phi.1$ and
can be located at two positions,
 one for all points in $\cI_{2i+1}$ and the other for all points in $\cI_{2i+2}$.
Given the online instance created in Step \ref{step:CD-7}, at each time step $i\leq r-1$, algorithm $\cA_{ATWC}$
creates 2 positions $x_{2i+1}$ and $x_{2i+2}$, for $\cI_{2i+1}$ and $\cI_{2i+2}$.


Let $T_{1},\cdots,T_{l}$ be the sequence of time intervals
sliced by the arriving times and departure times of all points in $S$.
For each $i\in [l]$, let $T_i = [left_i, right_i]$ and $|T_i| = right_i - left_i$ be the length of~$T_i$. 
Note that $left_1 = 0$ and $right_l = T$.
By construction, the set of present points only changes at the end-points of the intervals
and stays the same within each interval.
Accordingly, given the optimal offline positions $X_1,\dots, X_n\in P$,
$$OPT_C(S; P) = \int_0^T d_{min}(t; X_1,\dots, X_n) dt = \sum_{i\in[l]} |T_i| d_{min}(T_i; X_1, \dots, X_n),$$
where $d_{min}(T_i; X_1, \dots, X_n)$ is the minimum distance incurred by $X_1, \dots, X_n$ at any time $t$
within~$T_i$ (that is, $t\in (left_i, right_i)$).

Now we consider the minimum distance incurred by algorithm $\cA_{CD}$ in each $T_i$.
For each $i\in[l]$, let $S_{i} = \{j | j\in[n], s_j < right_i, d_j> left_i\}$ be the set of points that overlap with $T_{i}$,
and $n_{i}=|S_{i}|$.
Again by Claim \ref{claim:A_I} and in particular by property $\Phi.2$,
all points in $S_i$ are removed from $S$ in the first $n_i$ rounds of algorithm $\cA_{CD}$ and
$$S_i \subseteq \cI_1\cup \cI_2 \cup \dots\cup\cI_{2n_i}.$$
Accordingly,
all points in $S_i$ are located at the first $2n_i$ positions
created by algorithm $\cA_{ATWC}$ from time 0 to time $n_i-1$.
Since $\cA_{ATWC}$ is a $\sigma$-competitive algorithm,
it ensures that at time $n_i-1$ in the online instance,
the minimum distance
incurred by positions $x_1, \dots, x_{2n_i}$
is at least $\frac{Disp(2n_{i};P)}{\sigma}$.
Thus the cumulative distance in $T_i$  according to $\cA_{CD}$ is
$$|T_i|d_{min}(T_i; x_1, \dots, x_{2n_i}) \geq \frac{|T_i| Disp(2n_{i};P)}{\sigma},$$
and the total cumulative distance is
$$\int_0^T d_{min}(t; x_1,\dots, x_{2r}) dt = \sum_{i\in [l]} |T_i| d_{min}(T_i; x_1,\dots, x_{2n_i}) \geq \sum_{i\in [l]} \frac{|T_i| Disp(2n_{i};P)}{\sigma}.$$

Again because the set of present points stays the same throughout $T_i$, we have
$$d_{min}(T_i; X_1, \dots, X_n) \leq Disp(n_i; P)$$ and
$$OPT_C(S; P) \leq \sum_{i\in[l]} |T_i| Disp(n_i; P).$$
Therefore the competitive ratio is no larger than
$$\frac{\sigma \sum_{i\in[l]} |T_i| Disp(n_i; P)}{\sum_{i\in [l]} |T_i| Disp(2n_{i};P)}
\leq \sigma\cdot \max_{i\in [l]} \frac{Disp(n_i; P)}{Disp(2n_{i};P)}\leq \sigma \max\limits_{i\geq 1}\frac{Disp(i;P)}{Disp(2i;P)},$$
as we wanted to show.
Finally, all that remains is to prove the following claim.

\begin{claim}\label{claim:2i-i}
For any $k\geq 1$, $i\geq 1$ and $k$-dimensional polytope $P$,
$$\frac{Disp(i;P)}{Disp(2i;P)}\leq 2.$$
\end{claim}

\begin{proof}
Use $\{x_l\}_{l\in [i]}$ to denote the optimal positions for dispersing $i$ points in~$P$
and
let $\delta=\left(\frac{Disp(i;P)}{4},0,0,\dots, 0\right)$ be a $k$-dimensional vector.
Now we can define a set $X'$ of $2i$ positions,
$$X'=\{x'_j | x'_j = x_{\lceil j/2 \rceil} + \delta(-1)^{j \text{ mod }2}\}_{j\in [2i]}.$$
That is, $x'_{2\lceil j/2 \rceil-1}$ and $x'_{2\lceil j/2 \rceil}$ are obtained by shifting the first coordinate of
$x_{\lceil j/2 \rceil}$
by $\frac{Disp(i;P)}{4}$ in two directions.
By the definition of $X'$, for any $j$ and $j'$ in $[2i]$, 
\begin{itemize}
\item if ${\lceil j/2 \rceil} = {\lceil j'/2 \rceil}$ then
\begin{equation}\label{equ:CD-1}
dis(x'_j,x'_{j'})=\frac{Disp(i;P)}{2};
\end{equation}
\item if ${\lceil j/2 \rceil} \neq {\lceil j'/2 \rceil}$ then
\begin{equation}\label{equ:CD-2}
dis(x'_j,x'_{j'})\geq \frac{Disp(i;P)}{2}.
\end{equation}
\end{itemize}
Indeed, Equation \ref{equ:CD-1} is because $x'_j$ and $x'_{j'}$ only differ at their first coordinates by $\frac{Disp(i;P)}{2}$;
and Equation \ref{equ:CD-2} is by applying the triangle inequality twice and because $dis(x_{\lceil j/2 \rceil}, x_{\lceil j'/2 \rceil})\geq Disp(i; P)$:
\begin{eqnarray*}
& & dis(x_j',x_{j'}') \geq dis(x'_j,x_{\lceil j'/2 \rceil})-dis(x'_{j'},x_{\lceil j'/2 \rceil})
= dis(x'_j,x_{\lceil j'/2 \rceil})-\frac{Disp(i;P)}{4}\\
& \geq & dis(x_{\lceil j/2 \rceil}, x_{\lceil j'/2 \rceil})-
dis(x'_j,x_{\lceil j/2 \rceil})-\frac{Disp(i;P)}{4}\\
&=&  dis(x_{\lceil j/2 \rceil}, x_{\lceil j'/2 \rceil})- \frac{Disp(i;P)}{2}\geq \frac{Disp(i;P)}{2}.
\end{eqnarray*}

%
%
%
%
%
%
%
%


Moreover, for any $j\in[2i]$, since $x'_j$ differs from $x_{\lceil j/2 \rceil}$ only in the first coordinate by $\frac{Disp(i;P)}{4}$,
\begin{equation}\label{equ:CD-3}
dis(x'_j, \partial P) \geq dis(x_{\lceil j/2 \rceil}) - \frac{Disp(i;P)}{4} \geq Disp(i; P) - \frac{Disp(i; P)}{4} > \frac{Disp(i;P)}{2}.
\end{equation}
By Equations \ref{equ:CD-1}, \ref{equ:CD-2} and \ref{equ:CD-3}, and by the definition of $Disp(2i; P)$, we have
$$Disp(2i; P)\geq d_{min}(X') \geq \frac{Disp(i;P)}{2},$$
and Claim \ref{claim:2i-i} holds.
\end{proof}

By Claim \ref{claim:2i-i}, $\max\limits_{i\geq 1}\frac{Disp(i;P)}{Disp(2i;P)}\leq 2$ and Theorem \ref{thm:integral} holds.
\end{proof}

Combining Theorems \ref{thm:1-dtight}, \ref{thm:2d-worst}, \ref{thm:kd:approv} and \ref{thm:integral},
 we immediately have the following.

\begin{corollary}
For the CD problem,
there exists a polynomial-time offline algorithm that achieves competitive ratio arbitrarily close to $4\ln2$ in dimension 1;
competitive ratio $3.182$ for 2-dimensional squares;
and competitive ratio $\frac{4}{1-\epsilon}$ for any $k\geq 2$, $k$-dimensional polytope $P$ and $\epsilon>0$.
\end{corollary}

\section{Dispersion Without the Boundary Condition}
\label{appendix:scattering}
The literature has considered the dispersion problem when the distance to the boundary is not taken into consideration:
referred to as {\em spreading points} \cite{cabello2007approximation} or {\em facility dispersion} \cite{ravi1994heuristic}).
The objective is
$$SP(n; P) \triangleq \max_{X_1,\dots, X_n \in P}\ \min_{i, j\in [n]} dis(X_i, X_j).$$
In this section we consider the online dispersion problem for this objective.

Recall that for any instance $S = ((s_1, d_1), \dots, (s_n, d_n))$, $T = \max_{i\in [n]} d_i$.
Given locations $X = (X_1,\dots, X_n)$, for any $t\leq T$, let
$$d^{SP}_{min}(t; X) = \min_{i, j\in [n] : s_i\leq t \leq d_i, s_j\leq t\leq d_i} dis(X_i, X_j)$$
be the minimum distance among the points that are present at time $t$.
When $X$ is clear from the context, we may write $d^{SP}_{min}(t)$ for short.
Here we also consider two natural objectives:
the {\em all-time worst-case} (ATWC) problem, where the objective is
$$OPT^{SP}_A(S; P)  \triangleq \max_{X_1,\dots, X_n} \min_{t\leq T} d^{SP}_{min}(t);$$
and the {\em cumulative distance} (CD) problem, where the objective is
$$OPT^{SP}_C(S; P) \triangleq \max_{X_1, \dots, X_n} \int_0^T d^{SP}_{min}(t) dt.$$
Similar to Claim \ref{clm:atwc=disp},
letting $m = \max_{t\leq T} |\{i : s_i\leq t\leq d_i\}|$, we have $OPT^{SP}_A(S;P) = SP(m;P)$.

In the dispersion without boundary condition problem,
the optimal solution and the optimal competitive ratio of online algorithms may be
very different from the dispersion problem.
For instance, consider the polytope shown in Figure \ref{fig:scp}.
When the number of points is small,
the optimal dispersion without boundary condition locates most points in the left part of the polytope, while
the optimal dispersion locates most points in the right part.
\begin{figure}[htbp]
\begin{center}
\setlength{\unitlength}{0.6cm}
\thicklines
\begin{picture}(16,2)

\put(0, 0){\line(1,0){16}}
\put(0, 0.3){\line(1,0){14}}
\put(14, 2){\line(1,0){2}}
\put(14, 0.3){\line(0,1){1.7}}
\put(16, 0){\line(0,1){2}}
\put(0, 0){\line(0,1){0.3}}

\end{picture}
\caption{A $2$-dimensional polytope}
\label{fig:scp}
\end{center}
\end{figure}
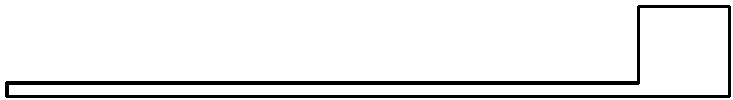

However, we show that most of our techniques for the dispersion problem
can be carried through to dispersion without boundary condition.
Indeed, for the 1-dimensional case, by locating the first two positions at the two end points of the segment,
we can lower- and upper-bound the optimal
competitive ratio for online ATWC as $2\ln 2$.
Because there is no boundary condition,
this lower-bound continues to hold for higher dimensions.
Indeed, we have the following.

\begin{theorem}
For any $k \geq 1$, no algorithm achieves a competitive ratio better
than $2\ln 2$ for online ATWC without boundary condition for arbitrary $k$-dimensional polytopes.
\end{theorem}
\begin{proof}
Consider a unit segment in a $k$-dimensional space. The lower-bound of $2\ln 2$ can be proved by directly applying Theorem \ref{thm:1d-lower}.
If we restrict ourselves to non-degenerate $k$-dimensional polytopes,
we can consider the hyperrectangle with radius $\epsilon$ in $k-1$ dimensions and radius 1 in one dimension.
It can be viewed as an approximation of a segment in $k$-dimensional space.
When $\epsilon$ is sufficiently small, $SP(n;P) \approx \frac{1}{n-1}$ for $n\geq 2$.
Thus, by applying a similar technique to the proof of Theorem \ref{thm:1d-lower},
the lower-bound of $2\ln 2$ holds.
\end{proof}

For squares, our technique in designing the online algorithm can still be applied, but the parameters need to be re-calculated. Accordingly, below we only show how to generalize the greedy algorithm for arbitrary polytopes to dispersion without boundary condition.

\begin{theorem}
For any $k\geq 2$, polytope $P$ with covering rate at least $\gamma>0$, and $\epsilon>0$,
there exists a deterministic polynomial-time algorithm for online ATWC without boundary condition,
with competitive ratio $\frac{2}{1-\epsilon}$ and running time polynomial in $\frac{1}{(\gamma\epsilon)^k}$.
\end{theorem}
\begin{proof}

We prove the competitive ratio $2$ for the greedy algorithm by induction.
First when there is only one point in $P$,
the optimal distance is equal to the maximal distance given by the greedy algorithm since the boundary is not taken into consideration.
Suppose when there are $n$ points in $P$ and the allocation of the $n$ points satisfies the competitive ratio $2$,
we find a feasible position for the $(n+1)th$ point satisfying the requirement.
Note that for arbitrarily fixed $n$ points $\{p_1,\dots,p_n\}$ where $p_i$ is in $P$ for $i\in [n]$, there exists a point $p^*$ such that
$$dis(p_i, p^*) \geq \frac{SP(n+1;P)}{2}, \forall i\in [n].$$
This is true because otherwise $n$ spheres centered at $\{p_1,\dots,p_n\}$ with radius $\frac{SP(n+1;P)}{2}$ is a covering for $P$.
Thus for $n+1$ points arbitrarily allocated in $P$,
there exists $i\in[n+1]$ such that there are at least two different points allocated in the same sphere centered at~$p_i$, excluding the boundary.
Therefore, the minimum distance among those two points, and hence the minimum distance among $n+1$ points is less than $SP(n+1;P)$.
Thus we derive a contradiction and the above inequality holds.
By allocating the $(n+1)$th point at position~$p^*$ gives us the~$2$-competitive ratio.

The method to approximate the greedy mechanism is similar to the proof of Theorem~\ref{thm:kd:approv}, but not exactly the same.
Note that when the boundary is not taken into consideration,
there may not exist a cube whose center is in $P$ and its distance to the greedy solution (denoted by $p^*$) is small.
However, the distance from $p^*$ to any point within the cube which contains $p^*$ is small.
By brute force searching the cube which intersects with $P$ and whose center has largest distance to the existing points,
allocating any point in the intersection of the cube and $P$ gives us a $\frac{2}{1-\epsilon}$-competitive algorithm with running time polynomial in $\frac{1}{(\gamma\epsilon)^k}$.
\end{proof}
%

Finally, the offline CD problem without boundary condition is similar to the original CD problem (see Section \ref{sec:kdcd} and Appendix \ref{appendix:offline:cd}),
and can be reduced to online ATWC without boundary condition with the competitive ratio scaling up by a factor of at most 2.

\section{The insert-only model}\label{app:insert_only}
Under the insert-only model
the online CD problem and the online ATWC problem are actually equivalent. On the one hand,
if an algorithm is $\sigma$-competitive for the online ATWC problem,
then directly applying the same algorithm achieves the same competitive ratio of $\sigma$
for the online CD problem. Indeed, since the algorithm is $\sigma$-competitive whenever a new point arrives,
it is actually $\sigma$-competitive within each time interval.
Taking integral over all time intervals, the competitive ratio $\sigma$ is preserved.

On the other hand, if an algorithm is $\sigma$-competitive for the online CD problem,
directly applying the same algorithm would be $\sigma$-competitive
for the online ATWC problem. Indeed, assuming otherwise,
 there must exist a point $i$ such that, in the time interval $[s_i, s_{i+1})$,
 the minimum distance $d_{min}(t)$ with $t\in [s_i, s_{i+1})$ 
 is smaller than a $\sigma$-fraction of $Disp(i; P)$.
So we can construct another instance $S'$ where the algorithm's competitive ratio for the online CD problem is violated.
More precisely, $S'$ keeps the arriving time of the first $i$ points unchanged, there is no other point arriving,
and 
the departure time $T$ of the $i$ points is set to be sufficiently large.
By doing so, the ratio between the cumulative distance and $d_{min}(t)T$ with $t\in [s_i, s_{i+1})$ is arbitrarily close to 1, and the ratio between the optimal cumulative distance and $Disp(i; P)T$ is arbitrarily close to 1.
Accordingly, the algorithm's competitive ratio for $S'$ for the online CD problem is worse than $\sigma$.%
\footnote{Intuitively $T =+\infty$, but it would be sufficient to set $T$ to be a large-enough finite number.}

Thus, all our algorithms for the online ATWC problem can be
directly applied to the online (and also offline) CD problem, with the competitive ratios unchanged.
Finally, all our inapproximability results for the online ATWC problem hold under the insert-only model.




\end{document}